\documentclass[12pt,english]{paper}
\usepackage[T1]{fontenc}
\usepackage[latin9]{inputenc}
\usepackage{geometry}
\usepackage{float}
\usepackage{amsthm}
\usepackage{amsmath}
\usepackage{amssymb}
\usepackage{graphicx}
\usepackage{esint}

\makeatletter

\providecommand{\tabularnewline}{\\}

\numberwithin{equation}{section}
\numberwithin{figure}{section}
\theoremstyle{plain}
\newtheorem{thm}{\protect\theoremname}

\@ifundefined{showcaptionsetup}{}{%
 \PassOptionsToPackage{caption=false}{subfig}}
\usepackage{subfig}
\makeatother

\usepackage{babel}
\providecommand{\theoremname}{Theorem}

\begin{document}

\title{Quantile Mechanics 3: Series Representations and Approximation of
some Quantile Functions appearing in Finance.}

\author{Asad Munir, William Shaw}

\date{19/03/2011}
\maketitle
\begin{abstract}
It has long been agreed by academics that the inversion method is
the method of choice for generating random variates, given the availability
of the quantile function. However for several probability distributions
arising in practice a satisfactory method of approximating these functions
is not available. The main focus of this paper will be to develop
Taylor and asymptotic series expansions for the quantile functions
belonging to the following probability distributions; Variance Gamma,
Generalized Inverse Gaussian, Hyperbolic and $\alpha$-Stable. As
a secondary matter, based on these analytic expressions we briefly
investigate the problem of approximating the quantile function. \end{abstract}

\section{Introduction}

Analytic expressions for quantile functions, have long been sought
after. The importance of these functions comes from their widespread
use in applications of statistics, probability theory, finance and
econometrics. Therefore much effort has been devoted into their study,
in particular since closed form expressions for the quantile function
of most distributions are not known, several approximations appear
in the literature. These approximations generally fall into one of
four categories, series expansions, functional approximations, numerical
algorithms or closed form expressions written in terms of a quantile
function of another distribution. The focus of this report is on the
former two categories.

For a given distribution function $F$ denote by $Q$ the associated
quantile function. We restrict our attention to the class of distributions
for which $F$ is strictly increasing and absolutely continuous. In
this case we have 

\begin{equation}
Q\left(u\right):=F^{-1}\left(u\right)\label{eq:QuantileEqualsInverseCDF}
\end{equation}
 where $F^{-1}$ is the compositional inverse of $F$. Suppose the
corresponding density function $f\left(x\right)$ is known. Differentiating
(\ref{eq:QuantileEqualsInverseCDF}) we obtain the \textit{first order
quantile equation},

\begin{equation}
\frac{dQ\left(u\right)}{du}=\frac{1}{F'\left(F^{-1}\left(u\right)\right)}=\frac{1}{f(Q\left(u\right))}.\label{eq:1stOrderQuantileODE}
\end{equation}
This is an autonomous equation in which the nonlinear terms are introduced
through the density function $f$. Many distributions of interest
have complicated densities, for example the densities of the generalized
hyperbolic distributions are written in terms of higher transcendental
functions. Thus the solutions to (\ref{eq:1stOrderQuantileODE}) are
often difficult to find, and hence this route has been relatively
unexplored in the literature. However provided the reciprocal of the
density $1/f$ is an analytic function at the initial condition $x_{0}=Q\left(u_{0}\right)$
we may employ some of the oldest methods of numerical integration
of initial value problems, namely the method of undetermined coefficients
and the method of successive differentiation, used to determine the
Taylor series expansions of $Q$. Since the equations in question
are non-linear, finding their solutions requires some special series
manipulation techniques found in for example \cite{hairer_solving_1993}.
Shaw and Steinbrecher \cite{steinbrecher_quantile_2008} considered
power series solutions to a nonlinear second order differential equation
focusing on particular distributions belonging to the Pearson family.
We will extend their work here by examining some non-Pearson distributions,
in particular we will look at the cases where $f$ is the density
function of the hyperbolic, variance gamma, generalized inverse Gaussian
and $\alpha$-stable distributions.

The central issue in developing algorithms to approximate the quantile
function $Q$ is that near the singular points the condition number
$\kappa_{Q}\left(u\right)$ defined by,

\[
\kappa_{Q}\left(u\right):=\left|\frac{uQ'\left(u\right)}{Q\left(u\right)}\right|
\]
is large, resulting in an ill conditioned problem (i.e. near the singular
points a small change in $u$ will result in large deviations in $Q$).
A remedy to this problem is to introduce a new variable $z$ to reduce
$\kappa_{Q}$. As an example consider the quantile function of the
standard normal distribution $\Phi^{-1}(u):=\sqrt{2}\textrm{erf}^{-1}\left(2u-1\right)$.
The associated condition number $\kappa_{\Phi^{-1}}\left(u\right)$
grows without bound as $u\rightarrow1$. Hence our problem is ill
conditioned near the singular point $u=1$, and one should therefore
seek a change of variable to reduce $\kappa_{\Phi^{-1}}\left(u\right)$,
for example by introducing the variable $z\left(u\right):=-\ln\left(1-u\right)$.
One may then consider the equivalent problem of approximating the
transformation $A(z):=\Phi^{-1}(1-e^{-z})$. This approach was first
taken in \cite{shaw_quantile_2010}, albeit from a slightly different
perspective. Of course in the normal case the transformed function
$A$ can be written down explicitly, since the form of $Q$ was know.
For situations in which this is not possible we formulate a first
order differential equation, which we will solve in the subsequent
sections for $A$. 

To motivate the idea suppose the \textit{base} distribution used in
the inversion method to generate random variates is the standard uniform
$U(0,1)$. Once a random variate from this distribution has been generated
we may generate a random variate $X$ from the \textit{target} distribution
with c.d.f. $F_{T}$ by setting $X=Q_{T}\left(U\right)$, where $Q_{T}$
is the associated quantile function of the target distribution and
$U\sim U\left(0,1\right)$. However as noted by Shaw and Brickman
in \cite{shaw_quantile_2010} the base distribution need not be restricted
to the standard uniform. Suppose we wish to generate random variates
from a target distribution with distribution function $F_{T}$ and
associated quantile function $Q_{T}$. Assume further that rather
than generating uniform random variates we can generate variates from
a base distribution with distribution function $F_{B}$ and associated
quantile function $Q_{B}$. Now if we could find a function $A$ which
associates the $u^{th}$ quantile from the base distribution with
the $u^{th}$ quantile from the target distribution, 

\begin{equation}
Q_{T}(u)=A\left(Q_{B}\left(u\right)\right),\label{eq:RecyclingFunctionDef}
\end{equation}
we could generate random variates from the target distribution simply
by setting $X=A\left(Q_{B}\left(U\right)\right).$ Note however the
function argument $Z=Q_{B}\left(U\right)$ is itself a random variate
from the base distribution. Thus we now have a recipe to generate
random variates from a specified target distribution given random
variates from a non-uniform base distribution.

There is no restriction on the choice of base distribution, however
for the purpose of approximating $Q_{T}$ a wise choice of the variable
$z:=Q_{B}\left(u\right)$ would be one that reduces the condition
number, $\kappa_{A}\left(z\left(u\right)\right)<\kappa_{Q}\left(u\right)$.
We will now derive an ordinary differential equation describing the
function $A$ arising from the change of variable. Starting with the
first order quantile equation for the target quantile function $Q_{T}$,
\begin{equation}
\frac{dQ_{T}\left(u\right)}{du}=\frac{1}{f_{T}(Q_{T}\left(u\right))},\label{eq:FirstOrderTargetQuantileEquation}
\end{equation}
we make a change of variable $z=Q_{B}\left(u\right)$. A simple application
of the chain rule for differentiation to (\ref{eq:RecyclingFunctionDef})
and using the fact that $Q_{B}$ also satisfies the first order quantile
equation gives,

\[
\frac{dQ_{T}\left(u\right)}{du}=\frac{dA}{dz}\frac{dz}{du}=\frac{dA}{dz}\frac{dQ_{B}}{du}=\frac{dA}{dz}\frac{1}{f_{B}(Q_{B}\left(u\right))}=\frac{dA}{dz}\frac{1}{f_{B}(z)}.
\]
Substituting this identity into (\ref{eq:FirstOrderTargetQuantileEquation})
we obtain a first order nonlinear differential equation, 

\begin{equation}
\frac{dA}{dz}=\frac{f_{B}\left(z\right)}{f_{T}(A\left(z\right))},\label{eq:RecyclingEquation}
\end{equation}
which we call the \textit{first order recycling equation.} A second
order version of this equation was treated in \cite{shaw_quantile_2010}. 

Note that the idea of expressing target quantiles in terms of base
quantiles is not a new one, indeed this is the idea behind the generalized
Cornish Fisher expansion, which originally introduced in \cite{cornish_moments_1938}
and \cite{fisher_percentile_1960} was generalized by Hill and Davis
in \cite{hill_generalized_1968}. The Cornish Fisher expansion has
many drawbacks however, we refer the interested reader to \cite{jaschke_cornish-fisher_2002}.
Another interesting idea along these lines was introduces by Takemura
\cite{takemura_orthogonal_1983}. Here the the Fourier series expansion
of the target quantile $Q_{T}$ is developed with respect to an orthonormal
basis of the form $\left\{ \psi_{i}\circ Q_{B}\right\} _{i=0}^{\infty}$,
where $\left\{ \psi_{i}\right\} _{i=0}^{\infty}$ is itself an orthonormal
basis for a set of square integrable functions, we refer the reader
to the original paper for details. Unlike the Cornish Fisher expansion
which is asymptotic in nature Takemura's approach yields a convergent
series in the $L^{2}$ norm. Note however the computation of the Fourier
coefficients usually requires numerical quadrature and that for approximation
purposes the $L^{\infty}$ norm is preferred. 

Standard numerical techniques such as root finding and interpolation
for approximating the associated quantile functions often fail, particularly
in the tail regions of the distribution, see for example \cite{derflinger_random_2010}.
Thus in addition to developing convergent power series solutions to
(\ref{eq:1stOrderQuantileODE}) and (\ref{eq:RecyclingFunctionDef})
we also develop asymptotic expansions of $Q$ at the singular points
$u=0$ and $u=1$. In general it is not easy to discover these asymptotic
behaviors, some trial and error and intelligent fiddling is required,
\cite[Ch. 4]{bender_advanced_1978}. Our approach is to consider $Q$
as being implicitly defined by the equation, 

\[
F(Q)=u.
\]
Asymptotic iteration methods, see for example \cite[Ch. 2]{bruijn_asymptotic_1981},
may then be used to obtain the leading terms in the expansion. For
numerical purposes the resulting expansions are divergent but may
be summed by employing an appropriate summation technique such as
Shank's $e_{k}$ or Levin's $u$ transforms. The reader is warned
that the approach taken throughout the paper is an applied one, the
reader will certainly notice a lack of rigor at certain points, in
particular in deriving these asymptotic expansions we will execute
some ``sleazy'' maneuvers. 

In this report we will focus on the variance gamma, generalized inverse
Gaussian, hyperbolic, normal inverse Gaussian and $\alpha$-stable
distributions. For each of these distributions we will find power
series expansions of $Q$ and $A$ as well as the asymptotic behavior
of $Q$ near its singularities. Despite the unsightly appearance of
the formulae given in the present paper for the coefficients appearing
in these expansions, they are simple to program if not tedious. It
is difficult to study analytically the convergence properties of these
series, hence we will proceed in an empirical manner. In the last
section of this report we experiment with various numerical algorithms
to approximate $Q$, observing that a scheme based on Chebyshev-Padé
approximants and optimally truncated asymptotic expansions seems most
useful.

\section{Hyperbolic Distribution}

Amongst others it was observed by \cite{eberlein_hyperbolic_1995}
and \cite{bingham_modelling_2001} the hyperbolic distribution is
superior to the normal distribution in modeling log returns of share
prices. Numerical inversion of the hyperbolic distribution function
was considered in \cite{leobacher_method_2002}. There numerical methods
were considered to solve (\ref{eq:1stOrderQuantileODE}) and only
the leading order behavior of the left and right tails was given.
Here we will provide an analytic solution to (\ref{eq:1stOrderQuantileODE})
and the full asymptotic behavior of the tails. The density of the
of the hyperbolic distribution $\textrm{Hyp}\left(\alpha,\beta,\delta,\mu\right)$
is given by, 

\begin{equation}
f\left(x;\alpha,\beta,\delta,\mu\right)=\frac{\gamma}{2\alpha\delta K_{1}\left(\delta\gamma\right)}e^{-\alpha\sqrt{\delta^{2}+(x-\mu){}^{2}}+\beta(x-\mu)},\label{eq:HyperbolicDensityZeroMean}
\end{equation}
where $\alpha>0$, $\left|\beta\right|<\alpha$, $\delta>0$ and $\mu\in\mathbb{R}$
are shape, skewness, scale and location parameters respectively and
for notational convenience we have set $\gamma=\sqrt{\alpha^{2}-\beta^{2}}$.
By defining $\alpha_{1}:=\delta\alpha$ and $\beta_{1}:=\delta\beta$,
we obtain an alternative parametrization in which $\alpha_{1}$ and
$\beta_{1}$ are now location and scale invariant parameters. Hence
without loss of generality we may set $\mu=0$ and $\delta=1$, since
$Q\left(u;\alpha_{1},\beta_{1},\mu,\delta\right)=\mu+\delta Q\left(u;\alpha_{1},\beta_{1},0,1\right)$.
The first order quantile equation (\ref{eq:1stOrderQuantileODE})
then reads, 

\begin{equation}
\frac{dQ}{du}=N_{0}e^{\alpha_{1}\sqrt{1+Q^{2}}-\beta_{1}Q}\label{eq:HyperbolicQE}
\end{equation}
where $N_{0}=2\alpha_{1}\text{K}_{1}\left(\gamma_{1}\right)/\gamma_{1}$
and $\gamma_{1}=\sqrt{\alpha_{1}^{2}-\beta_{1}^{2}}$. To form an
initial value problem (IVP) let $u_{0}\in\left(0,1\right)$ and impose
the initial condition $Q\left(u_{0}\right)=x_{0}$. For all practical
purposes $x_{0}$ is usually determined by solving the equation $F\left(x_{0}\right)-u_{0}=0$
using a root finding procedure. By applying the method of undetermined
coefficients we find $Q$ admits the Taylor series expansion, 

\[
Q\left(u\right)=\sum_{n=0}^{\infty}q_{n}(u-u_{0})^{n},
\]
where the coefficients $q_{n}$ are defined recursively as follows, 

\[
q_{n}\text{=}\begin{cases}
x_{\text{0}}, & n=0\\
\frac{N_{0}}{n}b_{n-1}, & n\geq1
\end{cases},
\]
\smallskip{}

\[
b_{n}=\begin{cases}
e^{\alpha_{1}\sqrt{1+x_{0}^{2}}-\beta_{1}x_{0}}, & n=0\\
\frac{1}{n}\sum_{k=1}^{n}k(\alpha_{1}a_{k}-\beta_{1}q_{k})b_{n-k}, & n\geq1
\end{cases},
\]
and,

\[
a_{n}=\begin{cases}
\sqrt{1+x_{0}^{2}} & n=0\\
\frac{1}{na_{0}}\left(nq_{n}q_{0}+\sum_{k=0}^{n-2}(k+1)(q_{k+1}q_{n-k-1}-a_{k+1}a_{n-k-1})\right) & n\geq1
\end{cases}.
\]

\medskip{}

To develop the asymptotic behavior at the singular points $u=0$ and
$u=1$ note that the hyperbolic distribution function $F$ satisfies
the relationship, 
\[
F\left(x;\alpha_{1},\beta_{1},1,0\right)=1-F\left(-x;\alpha_{1},-\beta_{1},1,0\right),
\]
which implies,

\[
Q\left(u;\alpha_{1},\beta_{1},1,0\right)=-Q\left(1-u;\alpha_{1},-\beta_{1},1,0\right).
\]
Hence without loss of generality we need only look for an asymptotic
expansion of $Q$ as $u\rightarrow1$. Now under the assumption $1\ll Q$
as $u\rightarrow1$, we obtain from (\ref{eq:HyperbolicQE}),

\[
\frac{dQ}{du}\sim N_{0}e^{\left(\alpha_{1}-\beta_{1}\right)Q},\quad\textrm{as}\; u\rightarrow1.
\]
Solving this asymptotic relationship along with the condition $Q\left(1\right)=\infty$
yields the leading order behavior of $Q$,

\[
Q\left(u\right)\sim-\frac{\ln\left(N_{0}\left(\alpha_{1}-\beta_{1}\right)\left(1-u\right)\right)}{\left(\alpha_{1}-\beta_{1}\right)},\quad\textrm{as}\; u\rightarrow1.
\]
From which we note that $Q$ has a logarithmic singularity at $u=1$.
To develop further terms in the expansion, write $x:=Q\left(u\right)$
from which it follows $u=F(x)$ and note that $Q$ is implicitly defined
by the equation, 

\[
1-F(x)=\int_{x}^{\infty}f(t)\, dt
\]
Rearranging and introducing the variable $v$ we obtain, 

\[
v\text{:=}(1-u)N_{0}=\int_{x}^{\infty}e^{-\alpha_{1}\sqrt{1+t^{2}}+\beta_{1}t}\, dt.
\]
The idea is to expand the integrand appearing in the right hand side
and integrate term-wise, this process can be carried out symbolically,
giving the first few terms in the expansion, 

\[
v\sim e^{-x(\alpha_{1}-\beta_{1})}\left(\frac{1}{\alpha_{1}-\beta_{1}}-\frac{\alpha_{1}}{2(\alpha_{1}-\beta_{1})x}+\frac{\alpha_{1}\left(4+\alpha_{1}(\alpha_{1}-\beta_{1})\delta^{2}\right)}{8(\alpha_{1}-\beta_{1})^{2}x^{2}}+O\left(\frac{1}{x^{3}}\right)\right).
\]
Taking logs and inverting the resulting series allows us to write
down the first few terms in the asymptotic expansion of $x$, 

\[
x\sim y+\frac{\alpha_{1}}{2\left(\alpha_{1}-\beta_{1}\right)y}+\frac{\alpha_{1}}{2\left(\alpha_{1}-\beta_{1}\right)^{2}y^{2}}+O\left(\frac{1}{y^{3}}\right)
\]
from which we can conjecture the form the asymptotic expansion of
$Q$ as, 

\begin{equation}
Q(u)\sim y+\sum_{n=1}^{\infty}\frac{q_{n}}{y^{n}},\quad\textrm{as}\; u\rightarrow1,\label{eq:AsympQHyperbolic}
\end{equation}
where,

\[
y=-\frac{\ln\left(N_{0}\left(\alpha_{1}-\beta_{1}\right)\left(1-u\right)\right)}{\left(\alpha_{1}-\beta_{1}\right)}.
\]
Substituting (\ref{eq:AsympQHyperbolic}) into the first order quantile
equation (\ref{eq:HyperbolicQE}) allows us to derive a recurrence
relationship for the coefficients $q_{n}$, 

\[
q_{n}\text{=}\begin{cases}
0 & n=0\\
-\frac{\alpha_{1}}{2(\alpha_{1}-\beta_{1})} & n=1\\
-\frac{1}{(\alpha_{1}-\beta_{1})}\left((n-1)q_{n-1}+\frac{1}{n}\sum_{k=1}^{n-1}kb_{k}c_{n-k}+\alpha_{1}d_{n}\right) & n\geq2
\end{cases},
\]
where,

\[
a_{n}\text{=}\begin{cases}
0 & n=0\\
\frac{1}{2}+q_{1} & n=1\\
q_{n}+d_{n} & n\geq2
\end{cases},
\]

\[
b_{n}=\alpha_{1}a_{n}-\beta_{1}q_{n},
\]

\[
c_{n}\text{=}\begin{cases}
1 & n=0\\
\frac{1}{n}\sum_{k=1}^{n}kb_{k}c_{n-k} & n\geq1
\end{cases},
\]
and,

\[
d_{n}=\frac{1}{n-1}\sum_{k=0}^{n-2}(k+1)(q_{k+1}q_{n-k-2}-a_{k+1}a_{n-k-2}).
\]

Note that (\ref{eq:AsympQHyperbolic}) is a divergent series, however
of the summation methods we tested we found Levin's $u$ transform
and Padé approximants particular useful for summing (\ref{eq:AsympQHyperbolic}).
In both cases analytic continuation was observed. Later we will briefly
look at algorithms for constructing rational approximants of $Q$
valid on the domain $\left[10^{-10},1-10^{-10}\right]$, but if necessary
one may utilize (\ref{eq:AsympQHyperbolic}) to obtain approximations
on a wider region. 

As mentioned in the introduction the problem of approximating the
quantile function near its singular points is an ill conditioned one.
To this end it will be useful to introduce a change of variable which
reduces the condition number $\kappa_{Q}$. Motivated by the asymptotic
behavior of $Q$ near its singular points, in particular its leading
order behavior we introduce the base distribution defined by the density,

\[
f_{B}\left(x\right):=\begin{cases}
p_{-}\left(\alpha_{1}+\beta_{1}\right)e^{\left(\alpha_{1}+\beta_{1}\right)x}, & x\leq x_{m}\\
p_{+}\left(\alpha_{1}-\beta_{1}\right)e^{-\left(\alpha_{1}-\beta_{1}\right)x}, & x>x_{m}
\end{cases}.
\]
Here $x_{m}:=\beta_{1}/\gamma_{1}$ is the mode of the hyperbolic
distribution, $p_{-}:=e^{-(\text{\ensuremath{\alpha}1}+\text{\ensuremath{\beta}1})x_{m}}p_{m}$,
$p_{+}:=e^{(\text{\ensuremath{\alpha}1}-\text{\ensuremath{\beta}1})x_{m}}\left(1-p_{m}\right)$
and $p_{m}=F_{T}\left(x_{m}\right)$. The associated distribution
and quantile functions can be written down as,

\[
F_{B}\left(x\right):=\begin{cases}
p_{-}e^{\left(\alpha_{1}+\beta_{1}\right)x}, & x\leq x_{m}\\
1-p_{+}e^{-\left(\alpha_{1}-\beta_{1}\right)x}, & x>x_{m}
\end{cases},
\]
and,

\[
Q_{B}\left(u\right)=\begin{cases}
\frac{1}{\alpha_{1}+\beta_{1}}\ln\left(\frac{u}{p_{-}}\right), & u\leq p_{m}\\
-\frac{1}{\alpha_{1}-\beta_{1}}\ln\left(\frac{1-u}{p_{+}}\right), & u>p_{m}
\end{cases},
\]
respectively. Substituting this choice of $f_{B}$ into the recycling
equation (\ref{eq:RecyclingEquation}) results in a left and right
problem, 

\[
\frac{dA}{dz}=p_{-}(\alpha_{1}+\beta_{1})Ce^{\alpha_{1}\sqrt{1+A^{2}}-\beta_{1}A+\left(\alpha_{1}+\beta_{1}\right)z},\quad z\leq x_{m},
\]
and

\[
\frac{dA}{dz}=p_{+}(\alpha_{1}-\beta_{1})Ce^{\alpha_{1}\sqrt{1+A^{2}}-\beta_{1}A-\left(\alpha_{1}-\beta_{1}\right)z},\quad z>x_{m},
\]
respectively, along with the suitably imposed initial conditions.
For the left problem, we choose $u_{0}\in\left(0,p_{m}\right]$ and
impose the initial condition $x_{0}=A\left(z_{0}\right)=Q_{T}\left(u_{0}\right)$,
where $z_{0}:=Q_{B}\left(u_{0}\right)$. Similarly for the right problem
choose $u_{0}\in\left[p_{-},1\right)$. Again through the method of
undetermined coefficients we obtain the Taylor series expansion of
$A$, 

\[
A\left(z\right)=\sum_{n=0}^{\infty}a_{n}(z-z_{0})^{n},
\]
where the coefficients are defined recursively by, 

\[
a_{n}=\begin{cases}
x_{0} & n=0\\
\frac{C}{n}\theta\left(\alpha_{1}+\beta_{1}\right)d_{n-1} & n\geq1
\end{cases},
\]

\[
b_{n}=\begin{cases}
\sqrt{1+x_{0}^{2}} & n=0\\
\frac{1}{nb_{0}}\left(na_{n}a_{0}+\sum_{k=0}^{n-2}(k+1)\left(a_{k+1}a_{n-k-1}-b_{k+1}b_{n-k-1}\right)\right) & n\geq1
\end{cases},
\]

\[
c_{n}=\begin{cases}
\alpha_{1}\left(b_{1}+\phi\right)+\beta_{1}\left(1-a_{1}\right) & n=1\\
\alpha_{1}b_{n}-\beta_{1}a_{n} & n\neq1
\end{cases},
\]
and,

\[
d_{n}=\begin{cases}
e^{\alpha_{1}\sqrt{1+x_{0}^{2}}-\beta_{1}x_{0}+\rho z_{0}} & n=0\\
\frac{1}{n}\sum_{k=1}^{n}kc_{k}d_{n-k} & n\geq1
\end{cases}.
\]
\smallskip{}
For the solution to the left and right problem make the following
replacements,

\begin{table}[H]

\caption{Coefficient Parameters}
\begin{center}

\begin{tabular}{|c|c|c|}
\hline 
 & Left & Right\tabularnewline
\hline 
\hline 
$\theta$ & $p_{-}$ & $p_{+}$\tabularnewline
\hline 
$\phi$ & $+1$ & $-1$\tabularnewline
\hline 
$\rho$ & $\left(\alpha_{1}+\beta_{1}\right)$ & $-\left(\alpha_{1}-\beta_{1}\right)$\tabularnewline
\hline 
\end{tabular}

\end{center}
\end{table}

\section{Variance Gamma}

The variance gamma distribution was introduced in the finance literature
by Madan and Seneta in \cite{madan_variance_1990}. To our knowledge
there has been very little written on the approximation of the variance
gamma quantile. The density is given by, 

\begin{equation}
f_{\textrm{VG}}\left(x;\lambda,\alpha,\beta,\mu\right)=\frac{\gamma^{2\lambda}}{\left(2\alpha\right)^{\lambda-1/2}\sqrt{\pi}\Gamma\left(\lambda\right)}\left|x-\mu\right|^{\lambda-\frac{1}{2}}\textrm{K}_{\lambda-\frac{1}{2}}\left(\alpha\left|x-\mu\right|\right)e^{\beta\left(x-\mu\right)},\label{eq:VarianceGammaDensity}
\end{equation}
where $\lambda>0$, $\alpha>0$, $\left|\beta\right|<\alpha$, $\gamma=\sqrt{\alpha-\beta}$
and $\mu\in\mathbb{R}$. Setting the location parameter $\mu$ to
zero and substituting the density (\ref{eq:VarianceGammaDensity})
into the quantile equation (\ref{eq:1stOrderQuantileODE}) gives, 

\begin{equation}
\frac{dQ}{du}=g(Q(u))\label{eq:VGQuantileEquation}
\end{equation}
where the function $g$ is defined as,

\[
g\left(y\right):=N_{0}\frac{e^{-\beta y}|y|^{\frac{1}{2}-\lambda}}{\textrm{K}_{\lambda-\frac{1}{2}}(\alpha|y|)}.
\]
and $N_{0}=\left(2\alpha\right)^{\lambda-1/2}\sqrt{\pi}\Gamma\left(\lambda\right)/\gamma^{2\lambda}$.
Our strategy to solve (\ref{eq:VGQuantileEquation}) will be to apply
the method of successive differentiation to obtain the Taylor series
representation of $Q$,

\begin{equation}
Q_{\textrm{VG}}\left(u\right)=\sum_{n=0}^{\infty}q_{n}(u-u_{0})^{n}\label{eq:VGTaylorQ}
\end{equation}
 where $q_{0}$ is determined by the imposed initial condition, and
the remaining coefficients are given by, 

\[
q_{n}=\left.\frac{d^{n-1}}{du^{n-1}}\right|_{u=u_{0}}g\left(Q\left(u\right)\right),\quad n\geq1
\]
Thus the problem reduces to finding the higher order derivatives of
the composition $g\circ Q$, which can be obtained recursively through
by applying Faà di Bruno's formula as follows. Note first that $g$
can be written as the product of three functions $a\left(y\right):=e^{-\beta y}$,
$b\left(y\right):=|y|^{\frac{1}{2}-\lambda}$ and $c\left(y\right):=\left[K_{\lambda-\frac{1}{2}}(\alpha|y|)\right]^{-1}$.
Thus an application of the general Leibniz rule yields, 

\begin{equation}
g^{\left(n\right)}\left(y\right)=N_{0}\sum_{k=0}^{n}\sum_{j=0}^{k}\binom{n}{k}\binom{k}{j}a^{\left(n-k\right)}\left(y\right)b^{\left(j\right)}\left(y\right)c^{\left(k-j\right)}\left(y\right).\label{eq:gVGDerivatives}
\end{equation}
The higher order derivatives of the functions $a$ and $b$ appearing
in (\ref{eq:gVGDerivatives}) are given by, 

\[
a^{(n)}\left(y\right)=\left(-\beta\right)^{n}e^{-\beta y},
\]

\[
b^{(n)}\left(y\right)=y^{\frac{1}{2}-n-\lambda}\prod_{i=0}^{n-1}\left(\frac{1}{2}-\lambda-i\right)\times\begin{cases}
+1 & y>0\\
-1 & y<0
\end{cases}.
\]
To find the $n^{th}$ derivative of $c$ note that $c$ can be written
as the composition $c\left(y\right)=c_{1}\circ c_{2}\left(y\right)$,
where $c_{1}\left(y\right):=1/y$ and $c_{2}\left(y\right):=K_{\lambda-\frac{1}{2}}(\alpha|y|)$.
Hence we may obtain $c^{\left(n\right)}\left(y\right)$ through an
application of Faà di Bruno's formula, 

\[
c^{\left(n\right)}\left(y\right)=\sum\frac{n!}{m_{1}!m_{2}!\cdots m_{n}!}c_{1}^{\left(m_{1}+\cdots+m_{n}\right)}\left(c_{2}\left(y\right)\right)\prod_{j=0}^{n}\left(\frac{c_{2}^{\left(j\right)}\left(y\right)}{j!}\right)^{m_{j}},
\]
where the summation is taken over all solutions $\left(m_{1},m_{2},\ldots,m_{n}\right)\in\mathbb{Z}_{\geq0}^{n}$
to the Diophantine equation,

\begin{equation}
m_{1}+2m_{2}+3m_{3}\cdots+nm_{n}=n.\label{eq:DiophantineEquation}
\end{equation}
The formulae for the higher order derivatives of $c_{1}$ and $c_{2}$
are given by, 

\[
c_{1}^{(n)}\left(y\right)=\left(-1\right)^{n}n!y^{-(n+1)},
\]
and

\[
c_{2}^{\left(n\right)}\left(y\right)=\left(-\frac{\alpha}{2}\right)^{n}\sum_{k=0}^{n}\binom{n}{k}\times\begin{cases}
\textrm{K}_{\lambda-\frac{1}{2}-(2k-n)}(\alpha y), & y>0\\
\left(-1\right)^{n}\textrm{K}_{\lambda-\frac{1}{2}-(2k-n)}(-\alpha y) & y<0
\end{cases}.
\]
Where we have used the identity \cite[eq. 10.29.5]{olver_nist_2010}, 

\begin{equation}
\textrm{K}_{v}^{\left(n\right)}\left(z\right)=\left(-\frac{1}{2}\right)^{n}\sum_{k=0}^{n}\binom{n}{k}\textrm{K}_{v-\left(2k-n\right)}\left(z\right),\label{eq:BesselKDerivatives}
\end{equation}
which can easily be proved by induction. Formula (\ref{eq:BesselKDerivatives})
shows that, to find the higher order derivatives with respect to $z$
of the modified Bessel function of the second kind $\textrm{K}_{v}^{\left(n\right)}\left(z\right)$
we need only a routine to compute $\textrm{K}_{v}\left(z\right)$.
Now given the scheme (\ref{eq:gVGDerivatives}) to compute $g^{\left(n\right)}\left(x\right)$
and the coefficient $q_{0}$ defined by the initial condition we may
then compute $q_{n}$ for $n\geq1$ recursively, by another application
of Faà di Bruno's formula. It is important to note that to generate
the coefficients $q_{n}$ appearing in (\ref{eq:VGTaylorQ}) does
not require the use of any symbolic computation. 

Next we will focus on deriving an asymptotic expansion for $Q\left(u\right)$.
Similar to the hyperbolic quantile the following equality holds, 

\[
Q\left(u;\lambda,\alpha,\beta,0\right)=-Q\left(1-u;\lambda,\alpha,-\beta,0\right),
\]
so again without loss of generality we need only seek an asymptotic
expansion of $Q$ as $u\rightarrow1$. Our strategy here will be to,
\begin{enumerate}
\item derive an asymptotic expansion for the density $f$ as $x\rightarrow\infty$,
\item integrate term-wise to obtain an asymptotic expansion for the distribution
function $F$ as $x\rightarrow\infty$,
\item and finally invert this expansion to obtain an asymptotic expansion
for the quantile function $Q$ as $u\rightarrow1$.
\end{enumerate}
For the first step we will make use of the asymptotic relationship
\cite[§10.40]{olver_nist_2010},

\[
\textrm{K}_{v}\left(z\right)\sim\left(\frac{\pi}{2z}\right)^{\frac{1}{2}}e^{-z}\sum_{k=0}^{\infty}\frac{a_{k}\left(v\right)}{z^{k}}\quad\textrm{as}\; z\rightarrow\infty,
\]
where,

\[
a_{k}\left(v\right):=\frac{1}{k!8^{k}}\prod_{j=1}^{k}\left(4v^{2}-(2j-1)^{2}\right).
\]
From which it follows

\[
f\left(x\right)\sim\frac{\gamma^{2\lambda}}{\left(2\alpha\right)^{\lambda}\Gamma\left(\lambda\right)}e^{-\left(\alpha-\beta\right)x}x^{\lambda-1}\sum_{k=0}^{\infty}\frac{a_{k}\left(\lambda-\frac{1}{2}\right)}{\alpha^{k}}x^{-k}\quad\textrm{as}\; x\rightarrow\infty.
\]
Working under the assumption that term-wise integration is a legal
operation we obtain an asymptotic expansion for the distribution function, 

\begin{eqnarray}
1-F\left(x\right) & \sim & \frac{\gamma^{2\lambda}}{2^{\lambda}\Gamma\left(\lambda\right)}\sum_{k=0}^{\infty}\alpha^{-\lambda-k}a_{k}\left(\lambda-2^{-1}\right)\int_{x}^{\infty}e^{-\left(\alpha-\beta\right)x}t^{\lambda-k-1}dt\nonumber \\
 & = & \frac{\gamma^{2\lambda}}{2^{\lambda}\Gamma\left(\lambda\right)}\sum_{k=0}^{\infty}\alpha^{-\lambda-k}a_{k}\left(\lambda-2^{-1}\right)\Gamma\left(\lambda-k,x\left(\alpha-\beta\right)\right),\label{eq:VGAsymptDeriv1}
\end{eqnarray}
where $\Gamma\left(a,z\right)$ is the upper incomplete gamma function,
which for large $z$ satisfies the asymptotic relationship \cite[§8.11]{olver_nist_2010},

\begin{equation}
\Gamma\left(a,z\right)\sim z^{a-1}e^{-z}\sum_{j=0}^{\infty}\frac{\Gamma\left(a\right)}{\Gamma\left(a-j\right)}z^{-j}.\label{eq:IncompleteGammaAsympExpansion}
\end{equation}
Substituting into (\ref{eq:VGAsymptDeriv1}) and assuming the terms
of the series may be rearranged we obtain,

\begin{equation}
1-F\left(x\right)\sim\frac{\left(2\alpha\right)^{-\lambda}\gamma^{2\lambda}}{\Gamma(\lambda)}x^{\lambda-1}e^{-x(\alpha-\beta)}\sum_{k=0}^{\infty}b_{k}x^{-k}\quad\textrm{as}\; x\rightarrow\infty,\label{eq:VGCDFAsympExpansion}
\end{equation}
where, 

\[
b_{k}=\sum_{j=0}^{k}(\alpha-\beta)^{-(j+1)}\alpha^{-(k-j)}\left(\prod_{i=0}^{k-1}(\lambda-k+i)\right)a_{k-j}\left(\lambda-\frac{1}{2}\right).
\]
The expression (\ref{eq:VGCDFAsympExpansion}) describes the asymptotic
behavior of the variance gamma distribution function as $x\rightarrow\infty$.
Let $u=F\left(x\right)$ in (\ref{eq:VGCDFAsympExpansion}), our goal
then is to invert this relationship to obtain an asymptotic expansion
of the quantile function $Q\left(u\right)$ as $u\rightarrow1$. Introducing
the variable, 

\[
v\text{:=}\frac{(2\alpha)^{\lambda}\sqrt{\pi}\Gamma(\lambda)}{\gamma^{2\lambda}}(1-u),
\]
and rearranging (\ref{eq:VGCDFAsympExpansion}) we obtain, 

\begin{equation}
v\sim x^{\lambda-1}e^{-x(\alpha-\beta)}D\left(\frac{1}{x}\right)\quad\textrm{as}\; x\rightarrow\infty.\label{eq:VGAsymptDeriv2}
\end{equation}
where $D$ is the formal power series defined by $D\left(z\right)=\sum_{k=0}^{\infty}b_{k}z^{k}$.
Taking logs and introducing the variable $y:=-\ln v/\left(\alpha-\beta\right)$,
we may write (\ref{eq:VGAsymptDeriv2}) as, 

\begin{equation}
x\sim y+\frac{\lambda-1}{\alpha-\beta}\log x+\frac{1}{\alpha-\beta}\log D\left(\frac{1}{x}\right)\label{eq:VGAsymptDeriv3}
\end{equation}
We wish to write $x$ in terms $y$; this task might at first may
appear difficult to achieve, but as it happens we are in luck, similar
expressions occur frequently in analytic number theory, and some useful
methods have been developed to invert these kinds of relationships.
A drawback of these methods is that they require symbolic computation.
The most basic method known to us, one can apply to invert (\ref{eq:VGAsymptDeriv3})
is the method of asymptotic iteration \cite{bruijn_asymptotic_1981}.
This method however is extremely slow. A much more efficient approach
is to use the method of Salvy\cite{de_recherche_asymptotic_1992}.
There it was noted that the form of the asymptotic inverse is given
by, 

\[
x=Q\left(u\right)\sim y+\sum_{n=0}^{\infty}\frac{P_{n}\left(\textrm{ln}y\right)}{y^{n}},
\]
where $P_{0}\left(\xi\right)$ is a polynomial of degree 1 and $P_{n}\left(\xi\right)$
are polynomials of degree $n$ for $n\geq1$. Following a similar
analysis of that in \cite{de_recherche_asymptotic_1992} it can by
shown that $P_{n}$ may be determined up to some unknown constant
terms $c_{0},\ldots,c_{n}$ by the recurrence relationship, 

\begin{equation}
P_{n}\left(\xi\right)=\begin{cases}
\frac{(\lambda-1)}{(\alpha-\beta)}\xi+c_{0} & n=0\\
\frac{(\lambda-1)}{(\alpha-\beta)}(P_{n-1}\left(\xi\right)-P_{n-1}\left(0\right))-\frac{(\lambda-1)}{(\alpha-\beta)}(n-1)\int_{0}^{\xi}P_{n-1}\left(t\right)\, dt+c_{n} & n\geq0
\end{cases}.\label{eq:VGAsymptDeriv4}
\end{equation}
The unknown terms $c_{0},\ldots,c_{n}$ may be computed through the
following iteration scheme,
\begin{itemize}
\item starting with $u_{0}\left(t\right)=\log\left(b_{0}\right)$, compute,
\[
u_{k}\left(t\right)=\frac{(\lambda-1)}{(\alpha-\beta)}\log\left(1+tu_{k-1}\left(t\right)\right)+\frac{1}{\alpha-\beta}\log D\left(\frac{t}{1+tu_{k-1}\left(t\right)}\right),\quad k=1,\ldots,n+1,
\]

\item extract the constants, 
\[
c_{k}=\left[t^{k}\right]u_{n+1}\left(t\right),\quad k=0,\ldots,n
\]
where $\left[t^{k}\right]u_{n+1}\left(t\right)$ is used to denote
the coefficient of the $t^{k}$ term in $u_{n+1}\left(t\right)$.
\end{itemize}
An implementation of this scheme in Mathematica code has been included
in appendix (TODO: Include code). Through this process the first few
terms of the asymptotic expansion of $Q$ may be generated as follows, 

\begin{eqnarray*}
Q\left(u\right) & \sim & y+\frac{(\lambda-1)\log y-\log\left(\alpha-\beta\right)}{\alpha-\beta}\\
 &  & \quad+\frac{(\lambda-1)\left(2\alpha+\alpha\lambda-\beta\lambda+2\alpha(\lambda-1)\log y-2\alpha\log\left(\alpha-\beta\right)\right)}{2y\alpha(\alpha-\beta)^{2}}+\cdots,\quad\textrm{as}\; u\rightarrow1.
\end{eqnarray*}

Next we focus on introducing a change of variable to reduce the condition
number $\kappa_{Q}$ near the tails and finding the Taylor series
representation of the corresponding function $A$. As in the hyperbolic
case motivated by the asymptotic behavior of the quantile function
near its singularities, in particular by its leading order behavior
we introduce the base distribution defined by the density,

\[
f_{B}\left(x\right):=\begin{cases}
p_{-}\left(\alpha+\beta\right)e^{\left(\alpha+\beta\right)x} & x\leq0\\
p_{+}\left(\alpha-\beta\right)e^{-\left(\alpha-\beta\right)x} & x>0
\end{cases},
\]
where $p_{-}=F\left(0\right)$ and $p_{+}=1-p_{-}$. The associated
distribution and quantile functions can be written down as,

\[
F_{B}\left(x\right):=\begin{cases}
p_{-}e^{\left(\alpha_{1}+\beta_{1}\right)x}, & x\leq0\\
1-p_{+}e^{-\left(\alpha_{1}-\beta_{1}\right)x}, & x>0
\end{cases},
\]
and,

\[
Q_{B}\left(u\right)=\begin{cases}
\frac{1}{\alpha_{1}+\beta_{1}}\ln\left(\frac{u}{p_{-}}\right), & u\leq p_{-}\\
-\frac{1}{\alpha_{1}-\beta_{1}}\ln\left(\frac{1-u}{p_{+}}\right), & u>p_{-}
\end{cases},
\]
respectively. Substituting this choice of $f_{B}$ into the recycling
equation (\ref{eq:RecyclingEquation}) results in a left and right
problem, 

\[
\frac{dA}{dz}=g_{L}(z),\quad z\leq Q_{B}\left(p_{-}\right),
\]
and

\[
\frac{dA}{dz}=g_{R}(z),\quad z>Q_{B}\left(p_{-}\right),
\]
respectively, along with the suitably imposed initial conditions.
For the left problem, we choose $u_{0}\in\left(0,p_{-}\right]$ and
impose the initial condition $x_{0}=A\left(z_{0}\right)=Q_{T}\left(u_{0}\right)$,
where $z_{0}:=Q_{B}\left(u_{0}\right)$. Similarly for the right problem
choose $u_{0}\in\left[p_{-},1\right)$. The functions $g_{L}$ and
$g_{R}$ appearing on right hand side of these differential equations
are defined as, 

\[
g_{\textrm{L}}\left(z\right):=p_{-}(\alpha+\beta)e^{(\alpha+\beta)z}g\left(A\left(z\right)\right),
\]
and,

\[
g_{\textrm{R}}\left(z\right):=p_{+}(\alpha-\beta)e^{-(\alpha-\beta)z}g\left(A\left(z\right)\right).
\]
Suppose that the series solution of either problem is given by, 

\[
A\left(z\right)=\sum_{n=0}^{\infty}a_{n}(z-z_{0})^{n}.
\]
Here the first coefficient $a_{0}$ is determined by the initial condition
imposed at $z_{0}$ and the remaining coefficients are given by, 

\[
a_{n}=g_{\textrm{L}}^{\left(n-1\right)}\left(z_{0}\right),\quad n\geq1,
\]
for the left problem and, 

\[
a_{n}=g_{\textrm{R}}^{\left(n-1\right)}\left(z_{0}\right),\quad n\geq1,
\]
for the right problem. Both sets of coefficients may easily be computed
from an application of Liebniz's rule,

\[
g_{\textrm{L}}^{\left(n\right)}\left(z\right)=p_{-}e^{(\alpha+\beta)z}\sum_{k=0}^{n}\binom{n}{k}(\alpha+\beta)^{k+1}\left(g\circ A\right)^{\left(n-k\right)}\left(z\right),
\]

\[
g_{\textrm{R}}^{\left(n\right)}\left(z\right)=p_{+}e^{-(\alpha-\beta)z}\sum_{k=0}^{n}\binom{n}{k}\left(-1\right)^{k}(\alpha-\beta)^{k+1}\left(g\circ A\right)^{\left(n-k\right)}\left(z\right).
\]
Starting with $a_{0}$, the higher order derivatives of the composition
$g\circ A$ appearing in these formulae, are computed recursively
in precisely the same way we computed the higher order derivatives
of $g\circ Q$ above.

\section{Generalized Inverse Gaussian}

The generalized inverse Gaussian (GIG) is a three parameter distribution,
special cases of which are the inverse Gaussian, positive hyperbolic
and Lévy distributions to name a few. It arises naturally in the context
of first passage times of a diffusion process. The probability density
function of a GIG random variable is given by, 

\begin{equation}
f_{\textrm{GIG}}\left(x;\lambda,\chi,\psi\right)=\frac{\left(\psi/\chi\right)^{\lambda/2}}{2\textrm{K}_{\lambda}\left(\sqrt{\psi\chi}\right)}x^{\lambda-1}e^{-\frac{1}{2}\left(\chi x^{-1}+\psi x\right)},\quad x>0,\label{eq:GIGDensity}
\end{equation}
where $\lambda\in\mathbb{R}$, $\chi>0$, $\psi>0$ and $\textrm{K}_{v}\left(z\right)$
is the modified Bessel function of the third kind with index $v$.
The GIG distribution is also used in the construction of an important
family of distributions called the generalized hyperbolic distributions;
more specifically a normal mean mixture distribution where the mixing
distribution is the GIG distribution results in a generalized hyperbolic
distribution. Consequently if one can generate GIG random variates
then a simple transformation may be applied to generate variates from
the generalized hyperbolic distribution \cite{weron_computationally_2004}. 

We will use an alternative parametrization to the standard one above,
let $\eta=\sqrt{\chi/\psi}$ and $\omega=\sqrt{\chi\psi}$, the density
then reads, 

\begin{equation}
f_{\textrm{GIG}}\left(x;\lambda,\eta,\omega\right)=\frac{1}{2\eta^{\lambda}\textrm{K}_{\lambda}\left(\omega\right)}x^{\lambda-1}e^{-\frac{\omega}{2}\left(\eta x^{-1}+\frac{1}{\eta}x\right)}.\label{eq:GIGDensityAlternativeParam}
\end{equation}
In this new parametrization $\omega$ and $\lambda$ are scale invariant
and $\eta$ is a scale parameter, so in the following without loss
of generality we may set $\eta=1$. The first order quantile equation
(\ref{eq:1stOrderQuantileODE}) now reads, 

\[
\frac{dQ}{du}=2\text{K}_{\lambda}\left(\omega\right)e^{\frac{1}{2}\omega\left(\frac{1}{Q}+Q\right)}Q{}^{1-\lambda}.
\]
Let $u_{0}\in\left(0,1\right)$ and impose the initial condition $Q\left(u_{0}\right)=x_{0}$.
For the case $\lambda\neq1$, the GIG quantile $Q$ admits the Taylor
series expansion, 

\begin{equation}
Q\left(u\right)=\sum_{n=0}^{\infty}q_{n}(u-u_{0})^{n},\label{eq:GIGTaylorSeries}
\end{equation}
where the coefficients $q_{n}$ are defined recursively as follows, 

\[
q_{n}=\begin{cases}
x_{0} & n=0\\
\frac{2}{n}K_{\lambda}(\omega)\sum_{i=0}^{n-1}b_{i}c_{n-i-1} & n\geq1
\end{cases},
\]

\[
a_{n}=\begin{cases}
\frac{1}{q_{0}} & n=0\\
-\frac{1}{q_{0}}\sum_{i=1}^{n}q_{i}a_{n-i} & n\geq1
\end{cases},
\]

\[
b_{n}=\begin{cases}
e^{\frac{\omega}{2}\left(a_{0}+q_{0}\right)} & n=0\\
\frac{\omega}{2n}\sum_{i=1}^{n}i\left(a_{i}+q_{i}\right)b_{n-i} & n\geq1
\end{cases},
\]
and

\[
c_{n}=\begin{cases}
q_{0}^{1-\lambda} & n=0\\
\frac{1}{q_{0}}\sum_{i=1}^{n}\left(\frac{(2-\lambda)i}{n}-1\right)q_{i}c_{n-i} & n\geq1
\end{cases}.
\]
For the special case $\lambda=1$, the coefficients are somewhat simplified,
with $a_{n}$ and $b_{n}$ as above the coefficients appearing in
(\ref{eq:GIGTaylorSeries}) become, 

\[
q_{n}=\begin{cases}
x_{0} & n=0\\
\frac{2}{n}K_{1}(\omega)b_{n-1} & n\geq1
\end{cases}.
\]
Next we will focus on developing the asymptotic behavior of $Q$ as
$u\rightarrow1$. We proceed in an analogous fashion to the variance
gamma case, and find that remarkably the form of asymptotic expansion
of $Q_{\textrm{GIG}}$ is very similar to that of $Q_{\textrm{VG}}$
as $u\rightarrow1$. From the definition of the distribution function
we have,

\[
1-F\left(x\right)=\frac{1}{2K_{\lambda}(\omega)}\int_{x}^{\infty}t^{\lambda-1}e^{-\frac{1}{2}\left(\frac{1}{t}+t\right)\omega}dt.
\]
Expanding the $e^{-\omega/2t}$ term and integrating term-wise we
obtain, 

\[
1-F\left(x\right)\sim\frac{1}{2K_{\lambda}(\omega)}\sum_{k=0}^{\infty}\frac{(-1)^{k}}{k!}\left(\frac{2}{\omega}\right)^{\lambda-2k}\Gamma\left(\lambda-k,\frac{\omega x}{2}\right),\quad\textrm{as}\; x\rightarrow\infty,
\]
where $\Gamma\left(a,z\right)$ is the upper incomplete gamma function.
As in the variance gamma case substituting (\ref{eq:IncompleteGammaAsympExpansion})
and rearranging the terms provides us with a more convenient form
of the asymptotic expansion for the generalized inverse Gaussian distribution
function,

\begin{equation}
1-F\left(x\right)\sim\frac{1}{2K_{\lambda}(\omega)}x^{\lambda-1}e^{-\frac{\omega x}{2}}\sum_{k=0}^{\infty}b_{k}x^{-k},\quad\textrm{as}\; x\rightarrow\infty,\label{eq:GIGCDFAsympExpansion}
\end{equation}
where,

\[
b_{k}=\sum_{j=0}^{k}\frac{(-1)^{k-j}}{(k-j)!}\left(\frac{\omega}{2}\right)^{k-2j-1}\left(\prod_{i=0}^{j-1}(\lambda-k+i)\right).
\]
Now let $u=F\left(x\right)$ and introduce the variable $v=2K_{\lambda}(\omega)(1-u)$,
then we can rewrite (\ref{eq:GIGCDFAsympExpansion}) as,

\begin{equation}
v\sim x^{\lambda-1}e^{-\frac{\omega x}{2}}D(\frac{1}{x}),\label{eq:GIGAsympDeriv1}
\end{equation}
where $D$ is the formal power series defined by $D\left(z\right)=\sum_{k=0}^{\infty}b_{k}z^{k}$.
To invert the asymptotic relationship (\ref{eq:GIGAsympDeriv1}) we
start by taking logs and introducing the variable $y:=-(2/\omega)\ln v$.
We may now write (\ref{eq:GIGAsympDeriv1}) as, 

\begin{equation}
x\sim y+\frac{2(\lambda-1)}{\omega}\log x+\frac{2}{\omega}\log D\left(\frac{1}{x}\right).\label{eq:GIGAsympDeriv2}
\end{equation}
One may now apply the method of asymptotic iteration \cite{bruijn_asymptotic_1981}
to invert (\ref{eq:GIGAsympDeriv2}). As mentioned earlier this method
however is extremely slow and a much more efficient approach is to
use the method of Salvy\cite{de_recherche_asymptotic_1992}. Again
the form of the asymptotic inverse is given by, 

\[
x=Q\left(u\right)\sim y+\sum_{n=0}^{\infty}\frac{P_{n}\left(\textrm{ln}y\right)}{y^{n}},
\]
where $P_{0}\left(\xi\right)$ is a polynomial of degree 1 and $P_{n}\left(\xi\right)$
are polynomials of degree $n$ for $n\geq1$ and following a similar
analysis of that in \cite{de_recherche_asymptotic_1992} it can by
shown that $P_{n}$ may be determined up to some unknown constant
terms $c_{0},\ldots,c_{n}$ by the recurrence relationship, 

\begin{equation}
P_{n}\left(\xi\right)=\begin{cases}
\frac{2(\lambda-1)}{\omega}\xi+c_{0} & n=0\\
\frac{2(\lambda-1)}{\omega}(P_{n-1}\left(\xi\right)-P_{n-1}\left(0\right))-\frac{2(\lambda-1)}{\omega}(n-1)\int_{0}^{\xi}P_{n-1}\left(t\right)\, dt+c_{n} & n\geq0
\end{cases}.\label{eq:GIGAsymptDeriv3}
\end{equation}
The unknown terms $c_{0},\ldots,c_{n}$ may be computed through the
following iteration scheme,
\begin{itemize}
\item starting with $u_{0}\left(t\right)=\log\left(b_{0}\right)$, compute,
\[
u_{k}\left(t\right)=\frac{2(\lambda-1)}{\omega}\log\left(1+tu_{k-1}\left(t\right)\right)+\frac{2}{\omega}\log D\left(\frac{t}{1+tu_{k-1}\left(t\right)}\right),\quad k=1,\ldots,n+1,
\]

\item extract the constants, 
\[
c_{k}=\left[t^{k}\right]u_{n+1}\left(t\right),\quad k=0,\ldots,n
\]
where $\left[t^{k}\right]u_{n+1}\left(t\right)$ is used to denote
the coefficient of the $t^{k}$ term in $u_{n+1}\left(t\right)$.
\end{itemize}
Through this process the first few terms of the asymptotic expansion
of $Q$ may be generated as follows, 

\begin{eqnarray*}
Q\left(u\right) & \sim & y+\frac{2\left((-1+\lambda)\ln y+\ln\left(\frac{2}{\omega}\right)\right)}{\omega}\\
 &  & \quad+\frac{4\lambda-\omega^{2}-4+4(-1+\lambda)^{2}\ln y+4(-1+\lambda)\ln\left(\frac{2}{\omega}\right)}{\omega^{2}y}+\cdots,\quad\textrm{as}\; u\rightarrow1.
\end{eqnarray*}
To observe the asymptotic behavior of $Q\left(u\right)$ as $u\rightarrow0$,
we utilize the following identity, 

\begin{equation}
Q\left(u;\lambda,1,\omega\right)=\frac{1}{Q\left(1-u;-\lambda,1,\omega\right)},\label{eq:GIGQuantileIdentity}
\end{equation}
which can be easily proven as follows; by definition of the density
(\ref{eq:GIGDensityAlternativeParam}) we have,

\[
f\left(x;\lambda,1,\omega\right)=\frac{1}{x^{2}}f\left(\frac{1}{x};-\lambda,1,\omega\right).
\]
Integrating both sides then yields,

\[
F\left(x;\lambda,1,\omega\right)=1-F\left(\frac{1}{x};-\lambda,1,\omega\right),
\]
from which (\ref{eq:GIGQuantileIdentity}) follows. 

Next we consider solving the recycling ODE, but first we must choose
a base distribution. Again motivated by the asymptotic behavior of
$Q$ as $u\rightarrow0$ and $u\rightarrow1$, in particular the leading
order behaviors we suggest the following base distribution characterized
by the density function, 

\begin{equation}
f_{B}\left(x\right)=\begin{cases}
p_{\textrm{L}}\frac{\omega}{2x^{2}}e^{-\frac{\omega}{2x}}, & x\leq x_{m}\\
p_{\textrm{R}}\frac{\omega}{2}e^{-\frac{\omega}{2}x}, & x>x_{m}
\end{cases}.\label{eq:GIGBaseDensity}
\end{equation}
Here $x_{m}$ serves as a cutoff point between two suitably weighted
density functions, in particular $x_{m}$ is the mode of the GIG distribution
defined by,

\[
x_{m}=\frac{\lambda-1+\sqrt{(\lambda-1)^{2}+\omega^{2}}}{\omega}.
\]
Note that the density function $f_{L}\left(x\right):=(\omega/2x^{2})e^{-\omega/2x}$
belongs to the scaled inverse $\chi^{2}$ distribution with $2$ degrees
of freedom and scale parameter $\omega/2$, and that the density function
$f_{R}\left(x\right):=(\omega/2)e^{-\omega x/2}$ is the density of
an exponential distribution with rate parameter $\omega/2$. The normalizing
constants $p_{\textrm{L}}$ and $p_{\textrm{R}}$ are defined by, 

\[
p_{\textrm{L}}=e^{\frac{\omega}{2x_{m}}}p_{m},
\]
and,

\[
p_{\textrm{R}}=e^{\frac{\omega}{2}x_{m}}\left(1-p_{m}\right),
\]
where $p_{m}=F_{\textrm{GIG}}\left(x_{m}\right).$ The associated
distribution and quantile functions can be written down as, 

\[
F_{B}\left(x\right)=\begin{cases}
p_{\textrm{L}}e^{-\frac{\omega}{2x}}, & x\leq x_{m}\\
1-\left(1-p_{\textrm{R}}\right)e^{-\frac{\omega}{2}\left(x-x_{m}\right)}, & x>x_{m}
\end{cases},
\]
and,

\begin{equation}
Q_{B}\left(u\right)=\begin{cases}
-\frac{\omega}{2\ln\left(u/p_{\textrm{L}}\right)} & u\leq p_{m}\\
x_{m}+\frac{2}{\omega}\ln\left(\frac{p_{m}-1}{u-1}\right) & u>p_{m}
\end{cases},\label{eq:GIGBaseQuantile}
\end{equation}
respectively. Substituting this choice of $f_{B}$ into the recycling
equation (\ref{eq:RecyclingEquation}) then leads to a left and right
problem given by, 

\begin{equation}
\frac{dA}{dz}=2p_{\textrm{L}}\text{K}_{\lambda}\left(\omega\right)\frac{\omega}{2z^{2}}e^{-\frac{\omega}{2z}}e^{\frac{1}{2}\omega\left(\frac{1}{A}+A\right)}A{}^{1-\lambda},\label{eq:GIGLeftRODE}
\end{equation}
and,

\begin{equation}
\frac{dA}{dz}=p_{\textrm{R}}\omega\text{K}_{\lambda}\left(\omega\right)e^{\frac{1}{2}\omega\left(\frac{1}{A}+A-z\right)}A{}^{1-\lambda},\label{eq:GIGRightRODE}
\end{equation}
respectively, along with the suitably imposed initial conditions.
For the left problem, we choose $u_{0}\in\left(0,p_{m}\right]$ and
impose the initial condition $x_{0}=A\left(z_{0}\right)=Q_{T}\left(u_{0}\right)$,
where $z_{0}:=Q_{B}\left(u_{0}\right)$. Similarly for the right problem
we choose $u_{0}\in\left[p_{m},1\right)$. Treating the left problem
first, we find,

\begin{equation}
A\left(z\right)=\sum_{n=0}^{\infty}a_{n}(z-z_{0})^{n},\label{eq:GIGRODESolution}
\end{equation}
where the coefficients are computed recursively through the identity,

\[
a_{n}\text{=}\begin{cases}
x_{0} & n=0\\
\frac{2p_{\textrm{L}}}{n}\text{K}_{\lambda}\left(\omega\right)\left(\sum_{k=0}^{n-1}\sum_{j=0}^{n-k-1}(k+1)b_{k+1}d_{j}e_{n-k-j-1}\right) & n\geq1
\end{cases},
\]
where,

\[
b_{n}\text{=}\begin{cases}
e^{-\omega/2z_{0}} & n=0\\
\frac{\omega}{2n}\sum_{k=0}^{n-1}(-1)^{k}\frac{(k+1)}{z_{0}^{k+2}}b_{n-k-1} & n\geq1
\end{cases},
\]

\[
c_{n}\text{=}\begin{cases}
\frac{1}{a_{0}} & n=0\\
-\frac{1}{a_{0}}\sum_{i=1}^{n}a_{i}c_{n-i} & n\geq1
\end{cases},
\]

\[
d_{n}\text{=}\begin{cases}
e^{\omega(c_{0}+a_{0})/2} & n=0\\
\frac{\omega}{2n}\sum_{i=1}^{n}i(c_{i}+a_{i})d_{n-i} & n\geq1
\end{cases},
\]
and

\[
e_{n}\text{=}\begin{cases}
a_{0}^{1-\lambda} & n=0\\
\frac{1}{a_{0}}\sum_{i=1}^{n}\left(\frac{(2-\lambda)i}{n}-1\right)a_{i}e_{n-i} & n\geq1
\end{cases}.
\]
The coefficients appearing in the series solution (\ref{eq:GIGRODESolution})
to the right problem \ref{eq:GIGRightRODE} are given by, 

\[
a_{n}\text{=}\begin{cases}
x_{0} & n=0\\
\frac{\omega p_{\textrm{R}}}{n}\text{K}_{\lambda}\left(\omega\right)\left(\sum_{k=0}^{n-1}d_{k}e_{n-k-1}\right) & n\geq1
\end{cases},
\]
where $c_{n}$ and $e_{n}$ are defined as in the solution to the
left problem above and,

\[
b_{n}\text{=}\begin{cases}
\frac{\omega}{2}(c_{1}+a_{1}-1) & n=1\\
\frac{\omega}{2}(c_{n}+a_{n}) & n\neq1
\end{cases},
\]
and,

\[
d_{n}\text{=}\begin{cases}
e^{b_{0}-\omega z_{0}/2} & n=0\\
\frac{1}{n}\sum_{i=1}^{n}ib_{i}d_{n-i} & n\geq1
\end{cases}.
\]

\section{$\alpha$-Stable}

Under the appropriate conditions Lagrange's inversion formula is capable
of providing us with a series representation of functional inverses.
Yet it seems to be ignored in the literature when one wants to find
an approximation to the quantile function, which itself is at least
in the continuous case, defined as the functional inverse of the c.d.f.
Based on this observation we provide a convergent series representation
for the quantile function of the asymmetric $\alpha$-stable distribution.
Note however the method is much more general than this; all it requires
is a power series representation of the c.d.f. 

Suppose the c.d.f. $F_{X}$ and quantile function $Q_{X}$ of a random
variable $X$ have the Taylor series representations, 

\[
F_{X}\left(x\right)=\sum_{n=1}^{\infty}\frac{f_{n}}{n!}\left(x-x_{0}\right)^{n},\;\textrm{and}\; Q_{X}\left(u\right)=\sum_{n=1}^{\infty}\frac{q_{n}}{n!}\left(u-u_{0}\right)^{n},
\]
respectively. The relationship between $F_{X}$ and $Q_{X}$ is given
by $F_{X}\left(Q_{X}\left(u\right)\right)=u$, for all $u\in\left(0,1\right)$.
The goal is to solve this expression, that is we would like to write
the Taylor coefficients $q_{n}$ in terms of $f_{n}$. One such expression
to achieve this is provided by Lagrange's Inversion formula written
in terms of Bell Polynomials, see \cite[§ 13.3]{aldrovandi_special_2001},

\begin{equation}
q_{n}=\begin{cases}
\frac{1}{f_{1}} & n=1\\
-\frac{1}{f_{1}^{n}}\sum_{k=1}^{n}q_{k}\mathbb{B}_{n,k}\left(f_{1},\ldots,f_{n-k+1}\right) & n\geq2
\end{cases},\label{eq:QuantileCoeffsLagrange}
\end{equation}
where the coefficients $\mathbb{B}_{n,k}\left(f_{1},\ldots,f_{n-k+1}\right)$
are the Bell polynomials. They are defined by a rather detailed expression, 

\begin{equation}
\mathbb{B}_{n,k}\left(f_{1},\ldots,f_{n-k+1}\right):=\sum_{\substack{v_{1},v_{2},\ldots\geq0\\
v_{1}+v_{2}+\cdots=k\\
v_{1}+2v_{2}+3v_{3}\cdots=n
}
}\frac{n!}{\prod_{j=1}^{n}\left[v_{j}!\left(j!\right)^{v_{j}}\right]}f_{1}^{v_{1}}f_{2}^{v_{2}}\cdots f_{n-k+1}^{v_{n-k+1}},\label{eq:BellPolynomials}
\end{equation}
where the summation is taken over all solutions $\left(v_{1},v_{2},\ldots,v_{n}\right)\in\mathbb{Z}_{\geq0}^{n}$
to the Diophantine equation,

\begin{equation}
v_{1}+2v_{2}+3v_{3}\cdots+nv_{n}=n\label{eq:DiophantineEquation-1}
\end{equation}
with the added constraint, the sum of the solutions is equal to $k$,
i.e.$\sum_{j=0}^{n}v_{j}=k$. For example the solutions $\left(v_{1},v_{2},v_{3},v_{4}\right)$
for $n=4$ are given by 

\[
\left\{ (4,0,0,0),\:(2,1,0,0),\:(0,2,0,0),\:(0,0,0,1),\:(1,0,1,0)\right\} .
\]
and if $k=2$ this picks out the solutions $(0,2,0,0)$ and $(1,0,1,0)$.
Note that solutions to (\ref{eq:DiophantineEquation-1}) correspond
exactly to the integer partitions of $n$. An integer partition of
a number $n$ is an unordered sequence of positive integers who's
sum is equal to $n$. The added constraint implies we should look
for partitions in which the number of non zero summands is equal to
$k$. For example the integer partitions of $n=4$ are given by $(1,1,1,1)$,
$(1,1,2)$, $(2,2)$, $(4)$, and $(1,3)$, and in the case $k=2$
this singles out the partitions $(2,2)$ and $(1,3)$. Thus to summarize
the sum in (\ref{eq:BellPolynomials}) is taken over all integer partitions
of $n$ in which the number of summands is given by $k$. 

The above recursion (\ref{eq:QuantileCoeffsLagrange}) can be solved
(see \cite[§ 13.3]{aldrovandi_special_2001}) leading to a more computationally
efficient expression for the coefficients, 

\begin{equation}
q_{n}=\begin{cases}
\frac{1}{f_{1}} & n=1\\
\frac{1}{f_{1}^{n}}\sum_{k=1}^{n-1}(-1)^{k}\frac{(n+k-1)!}{(n-1)!}\mathbb{B}_{n-1,k}\left(\frac{f_{2}}{2f_{1}},\frac{f_{3}}{3f_{1}},\ldots,\frac{f_{n-k+1}}{(n-k+1)f_{1}}\right) & n\geq2
\end{cases}.\label{eq:QuantileCoeffsLagrangeRecSolved}
\end{equation}

The $\alpha$-stable distribution, denoted $\textrm{S}_{\alpha}\left(\beta,\mu,\sigma\right)$
is commonly characterized through its characteristic function. However
there are many parametrizations which have lead to much confusion.
Thus from the outset we state explicitly the three parametrizations
we will work with in this report and provide the relationships between
them. We will call these parametrizations$\textrm{P}_{0}$, $\textrm{P}_{1}$
and $\textrm{P}_{2}$ respectively. In the following a subscript under
a parameter denotes the parametrization being used. We will primarily
work with Zoltarev's type (B) parametrization, see \cite[p. 12]{zolotarev_one-dimensional_1986}
denoted by $\textrm{P}_{2}$. In this parametrization the characteristic
function takes the form, 

\begin{equation}
\phi\left(t\right)=\exp\left\{ \sigma_{2}\left(it\mu_{2}-\left|t\right|^{\alpha}e^{-i(\pi/2)\beta_{2}K(\alpha)\text{sgn}\left(t\right)}\right)\right\} ,\label{eq:StableCFParam2}
\end{equation}
where $\alpha\in\left(0,2\right]$ is the tail index, $\mu_{2}\in\mathbb{R}$
is a location parameter, $\sigma_{2}>0$ is a scale parameter, $\beta\in\left[-1,1\right]$
is an asymmetry parameter and $K(\alpha):=\alpha-1+\text{sgn}(1-\alpha)$.
$\textrm{P}_{1}$ is the classic parametrization, and is probably
the most common due to the simplicity of the characteristic function
given by, see \cite[p. 5]{samorodnitsky_stable_1994},

\begin{equation}
\phi\left(t\right)=\exp\left\{ -\sigma_{1}^{\alpha}\left|t\right|^{\alpha}\left(1-i\beta_{1}\text{tan}\left(\frac{\pi\alpha}{2}\right)\text{sgn}\left(t\right)\right)+i\mu_{1}t\right\} .\label{eq:StableCFParam1}
\end{equation}
A connection between the parametrizations $\textrm{P}_{2}$ and $\textrm{P}_{1}$
can be derived by taking logarithms and equating first the real parts
of (\ref{eq:StableCFParam2}) and (\ref{eq:StableCFParam1}), followed
by the coefficients of $t\text{ and }\left|t\right|^{\alpha}$ in
the imaginary parts, leading to the set of relations, 

\begin{eqnarray*}
\mu_{1} & = & \mu_{2}\sigma_{2}\\
\sigma_{1} & = & \left(\text{cos}\left(\frac{1}{2}\pi K\left(\alpha\right)\beta_{2}\right)\sigma_{2}\right){}^{\frac{1}{\alpha}}\\
\beta_{1} & = & \text{cot}\left(\frac{\pi\alpha}{2}\right)\text{tan}\left(\frac{1}{2}\pi K\left(\alpha\right)\beta_{2}\right).
\end{eqnarray*}
Despite the fact no closed form expression for the c.d.f. $F_{\alpha}\left(x;\beta,\mu,\sigma\right)$
of the stable distribution in the general case is known, it can be
expressed in terms of an infinite series expansion. As is usual for
location-scale families of distribution, without loss of generality
we may set the location $\mu_{2}$ and scale $\sigma_{2}$ parameters
to 0 and 1 respectively. In addition it is sufficient to consider
expansions of the c.d.f. $F_{\alpha}\left(x;\beta,0,1\right)$ for
values of $x>0$ only since the following equality holds, 

\begin{equation}
F_{\alpha}\left(x;\beta,0,1\right)=1-F_{\alpha}\left(-x;-\beta,0,1\right).\label{eq:StableCDFXSymmetry}
\end{equation}

\begin{thm}
\label{thm:StableCDFSeries}Let $X\sim\textrm{S}_{\alpha}\left(x;\beta_{2},0,1\right)$
be a standard $\alpha$-stable random variable. Then the cumulative
distribution function of $X$ admits the following infinite series
representations,
\begin{equation}
F_{\alpha}\left(x;\beta_{2},0,1\right)=\sum_{n=0}^{\infty}\frac{f_{n}}{n!}x^{n},\label{eq:StableCDFExpansion1}
\end{equation}
where, 
\[
f_{n}=\begin{cases}
\frac{1}{2}\left(1-\frac{\beta_{2}K\left(\alpha\right)}{\alpha}\right) & n=0\\
(-1)^{n-1}\frac{1}{\pi}\frac{\Gamma\left(\frac{n}{\alpha}+1\right)}{n}\text{sin}\left(\frac{\pi n}{2}\left(1+\frac{\beta_{2}K\left(\alpha\right)}{\alpha}\right)\right) & n\geq1
\end{cases},
\]
and
\begin{equation}
F_{\alpha}\left(x;\beta_{2},0,1\right)=\sum_{n=0}^{\infty}\frac{\tilde{f}_{n}}{n!}x^{-\alpha n},\label{eq:StableCDFExpansion2}
\end{equation}
where, 
\[
\tilde{f}_{n}\text{=}\begin{cases}
1 & n=0\\
(-1)^{n}\frac{1}{\alpha\pi}\frac{\Gamma\left(\alpha n+1\right)}{n}\text{sin}\left(\frac{n\pi}{2}(\alpha+\beta_{2}K\left(\alpha\right))\right) & n\geq1
\end{cases},
\]
If $1<\alpha\leq2$ then (\ref{eq:StableCDFExpansion1}) is absolutely
convergent for all $x>0$ and if $0<\alpha<1$ then (\ref{eq:StableCDFExpansion1})
is an asymptotic expansion of $F_{\alpha}\left(x;\beta_{2},0,1\right)$
as $x\rightarrow0$. In contrast if $0<\alpha<1$ then the series
(\ref{eq:StableCDFExpansion2}) is absolutely convergent for $x>0$
and if $1<\alpha\leq2$ then (\ref{eq:StableCDFExpansion2}) is an
asymptotic expansion of $F_{\alpha}\left(x;\beta_{2},0,1\right)$
as $x\rightarrow\infty$. \end{thm}
\begin{proof}
The proof proceeds by obtaining first a series expansion of the density
of $X$. This is achieved by applying the inverse Fourier transform
to (\ref{eq:StableCFParam2}), expanding the exponential function,
and then performing a contour integration, see \cite[§ 5.8]{lukacs_characteristic_1970}
for details. The expansion in the density can be integrated term by
term to obtain an expansion of the c.d.f. 
\end{proof}
Note when $1<\alpha\leq2$, the series (\ref{eq:StableCDFExpansion1})
rapidly converges for values of $x$ near zero, where as (\ref{eq:StableCDFExpansion2})
converges rapidly for large values of $x$. The opposite is true when
$0<\alpha<1$. Thus we can now apply Lagrange's inversion formula
to find a series expansion of the quantile function $Q_{\alpha}$
in the central and tail regions. Note however that since the expansions
(\ref{eq:StableCDFExpansion1}) and (\ref{eq:StableCDFExpansion2})
are valid only for $x>0$, the resulting expansions of $Q_{\alpha}$
are only valid for $u>u_{0}$, where $u_{0}$ is the zero quantile
location defined by $u_{0}:=F_{\alpha}\left(0;\beta_{2},0,1\right)$.
This does not pose a restriction however since it follows from (\ref{eq:StableCDFXSymmetry}), 

\[
Q_{\alpha}\left(u;\beta,0,1\right)=-Q_{\alpha}\left(1-u;-\beta,0,1\right).
\]
Applying Lagrange's inversion formula to (\ref{eq:StableCDFExpansion1})
we see that the quantile function has the following infinite series
representation valid for $u>u_{0}$, 

\begin{equation}
Q_{\alpha}(u;\beta,0,1)=\sum_{n=1}^{\infty}\frac{q_{n}}{n!}\left(u-u_{0}\right){}^{n},\label{eq:StableQuantileSeries1}
\end{equation}
where the coefficients $q_{n}$ are given by (\ref{eq:QuantileCoeffsLagrangeRecSolved})
and 

\[
u_{0}=\frac{1}{2}\left(1-\frac{\beta_{2}K\left(\alpha\right)}{\alpha}\right).
\]
To find the functional inverse of (\ref{eq:StableCDFExpansion2})
we make the change of variable $y:=x^{\frac{1}{\alpha}}$, and apply
Lagrange's inversion formula to the power series, 

\[
G\left(y\right):=\sum_{n=0}^{\infty}\frac{\tilde{f}{}_{n}}{n!}y^{n}.
\]
The quantile function is then given by, 

\[
Q_{\alpha}(u;\beta,0,1)=\left[G^{-1}\left(u\right)\right]^{-\frac{1}{\alpha}},
\]
where,

\begin{equation}
G^{-1}\left(u\right)=\sum_{n=1}^{\infty}\frac{\tilde{q}_{n}}{n!}\left(u-1\right){}^{n},\label{eq:StableQuantileSeries2}
\end{equation}
and the coefficients $\tilde{q}_{n}$ are given by (\ref{eq:QuantileCoeffsLagrangeRecSolved})
with $f_{n}$ replaced by $\tilde{f}_{n}$. Note when $1<\alpha\leq2$,
the series (\ref{eq:StableQuantileSeries1}) rapidly converges for
values of $u$ near $u_{0}$, where as (\ref{eq:StableQuantileSeries2})
converges rapidly for values of $u$ close to $1$. In this case partial
sums of (\ref{eq:StableQuantileSeries1}) serve as good approximations
of the quantile function in the central regions where as partial sums
of (\ref{eq:StableQuantileSeries2}) can be used to approximate the
tails. The opposite is true when $0<\alpha<1$. The first few terms
of (\ref{eq:StableQuantileSeries1}) are given by, 

\[
Q_{\alpha}(u;\beta,0,1)=\frac{\pi\text{csc}\left(\pi\rho\right)}{\Gamma\left(1+\frac{1}{\alpha}\right)}\left(u-u_{0}\right)+\frac{\pi^{2}\text{cot}\left(\pi\rho\right)\text{csc}\left(\pi\rho\right)\Gamma\left(\frac{2+\alpha}{\alpha}\right)}{2\Gamma\left(1+\frac{1}{\alpha}\right)^{3}}\left(u-u_{0}\right){}^{2}+\textrm{O}\left(\left(u-u_{0}\right)^{3}\right)
\]

Implementation is a rather straightforward matter, high level languages
such as Mathematica have a built in implementation of the Bell polynomials
and plenty of algorithms exist to generate integer partitions in lower
level languages such as C++, see for example \cite{zoghbi_fast_1998}.
As long as we have a series representation of the c.d.f. using this
approach we could derive a series representation for the associated
Quantile function. However there is one obvious drawback, even though
the coefficients $q_{n}$ can be computed using elementary algebraic
operations, the number of partitions $p\left(n\right)$ of an integer
$n$ grows exponentially with $n$, thus for large values of $n$
the sum in (\ref{eq:BellPolynomials}) may be computationally expensive
due to the large number of summands. 

However the problem of reverting a power series is a classical one
in mathematics and many efficient algorithms have been devised as
a result. For example Knuth \cite[§ 4.7]{knuth_art_1998} gives several
algorithms for power series reversion including a classical algorithm
due to Lagrange (1768) that requires $O\left(N^{3}\right)$ operations
to compute the first $N$ terms. More recently Brent and Kung \cite{brent_fast_1978}
provide an algorithm which requires only about $150\left(N\log N\right)^{3/2}$
floating point operations. Dahlquist et al. \cite{dahlquist_numerical_2008}
also present a convenient but slightly less efficient algorithm based
on Toeplitz matrices. 

Concerning the numerical evaluation of the distribution function $F_{\alpha}$
of the stable distribution it has been remarked by various authors
\cite{stoyanov_numerical_2004} that the series expansions given in
\ref{thm:StableCDFSeries} are only useful for approximating $F_{\alpha}$
for either small or large values of $x$, due to the slow convergence
of the series. It is for this reason standard methods such as the
Fast Fourier transform or numerical quadrature techniques are applied
to evaluate $F_{\alpha}$. However we found in our experimentation
that series acceleration techniques such as Padé approximants and
Levin type transforms \cite{weniger_nonlinear_2003} could be applied
to (\ref{eq:StableCDFExpansion1}) and (\ref{eq:StableCDFExpansion2}).
The resulting rational functions were affective for approximation
purposes. The same comments apply for the series representation of
the density and quantile functions, at least for the set of parameters
we tested, which include those occurring frequently in financial data.
We will discuss numerical issues further in the next section.

\section{Numerical Techniques and Examples}

The goal of this section is to discuss the design of an algorithm
which accepts a set of distribution parameters and constructs at runtime
an approximation $Q_{A}$ to the quantile function satisfying some
prespecified accuracy requirements. In particular we are interested
in parameters occurring most often in Financial data, but the proposed
algorithm works for a much wider range of the parameter space. For
the distributions we have considered in this paper, most published
algorithms of this nature are based on interpolating or root finding
techniques. We break this mold by briefly examining other methods
of approximation based on certain convergence acceleration techniques.
We shall call the time it takes to construct an approximation the
setup time, and the time it takes to evaluate $Q_{A}$ at a point
$u$ the execution time of the algorithm. With the analytic expressions
made available in this report, a wide range of possibilities become
available. For instance the expansions may be used in conjunction
with other numerical techniques: 
\begin{itemize}
\item Techniques based on interpolation often fail in the tail regions \cite{derflinger_efficient_2009},
for this reason a good idea would be to supplement the algorithm with
the asymptotic expansions developed above. 
\item Root finding techniques are known to converge slowly. To improve the
rate of convergence one needs to supply a good initial guess of the
root. Such as provided by the truncated Taylor series provided above,
or better yet a Padé approximant. 
\end{itemize}
Another plausible approach is to construct a numerical integrator
based on the Taylor method. It has been reported by many authors that
when high precision is required this is the method of choice, see
for example \cite{corliss_solving_1982}, \cite{barrio_vsvo_2005}
and \cite{jorba_software_2005}. The idea here is to discretize the
domain $\left[a,b\right]\subset\left(0,1\right)$ into a non-uniform
grid $a=u_{0},\ldots,u_{n}=b$. To build this grid we must determine
the step sizes $h_{k}=u_{k}-u_{k-1}$ for $k=1,\ldots n$. In addition
at each grid point $u_{k}$ we must determine the order $m_{k}$ of
Taylor polynomial so that the required accuracy goals are achieved.
That is both the stepsize and order are variable. Based on two or
more terms of the Taylor series certain tests have been devised to
compute the $h_{k}$ and $m_{k}$, see again the references mentioned
above and the review article \cite{halin_applicability_1983}. 

The result of Taylor's method is a piecewise polynomial approximation
to the quantile function. However in this report we are more interested
in constructing rational function approximants of the form,

\[
R_{m,n}\left(v\right)=\frac{P_{m}\left(v\right)}{Q_{n}\left(v\right)}=\frac{\sum_{i=0}^{m}a_{i}v^{i}}{\sum_{i=0}^{n}b_{i}v^{i}}.
\]
and so we will not discuss Taylor's method further. Some of the best
known algorithms for approximating quantile functions are based on
rational function approximations. Unfortunately such approximations
traditionally are only available for ``simple'' distributions, see
for example \cite{moro_full_1995}, \cite{shaw_quantile_2010} or
\cite{acklam_algorithm_2009}. For more ``complicated'' distributions
alternative techniques are usually employed such as root finding or
interpolation. As mentioned above, these methods have severe limitations,
such as slow execution or setup times. An even bigger setback with
these techniques however is their failure in the tail regions. The
goal of this section is to devise an algorithm to overcome these problems. 

Ideally one would like to construct the best rational approximation
$R_{m,n}^{*}\left(v\right)$ to the quantile function $Q$ in the
minimax sense. There are numerous methods available such as the second
algorithm of Remes \cite{ralston_rational_1965} for the construction
$R_{m,n}^{*}\left(v\right)$. However these methods require several
a large number of function evaluations, and since $Q$ may only be
evaluated to high precision through a slow root finding scheme such
algorithms usually lead to unacceptable setup times. 

Therefore we suggest four alternative algorithms based on certain
series acceleration techniques %
\footnote{All of which have been prototyped in Mathematica using double precision
arithmetic, and do not rely on any of Mathematica's symbolic capabilities,
making them portable to lower level languages such as C++ or Fortran. %
}. The procedures assume the availability of an integration and root
finding routine to compute the distribution and quantile functions
respectively to full machine precision. These routines will be used
to compute the initial conditions and manage the error in the approximation
of $Q$. The inputs to the algorithm are: 1) the distribution parameters,
2) the required accuracy $\epsilon$. We partition the unit interval
and for convenience name each part as follows, 

\begin{eqnarray*}
\left(0,\tau_{L}\right) & = & \textrm{Left Tail Region}\\
\left[\tau_{L},u_{1}\right) & = & \textrm{Left Region}\\
\left[u_{1},u_{2}\right] & = & \textrm{Central Region}\\
\left(u_{2},\tau_{R}\right] & = & \textrm{Right Region}\\
\left(1-\tau_{R},1\right) & = & \textrm{Right Tail Region}
\end{eqnarray*}
The idea behind the first algorithm is then as follows, 
\begin{itemize}
\item The asymptotic expansions of $Q$ as developed above are employed
to approximate $Q$ in the left and right tail regions. Of course
these asymptotic expansions are divergent; however we have found constructing
Padé approximants and Levin-type sequence transforms particularly
useful in summing the series. The order of the approximant is chosen
in advance, and a numerical search is conducted to determine $\tau_{L}$
and $\tau_{R}$, which are typically small values $\thickapprox10^{-9}$. 
\item We will choose the mode quantile location $u_{m}:=F\left(x_{m}\right)$,
where $x_{m}$ is the mode of the distribution as the midpoint of
the central region. On this region the Taylor expansion of $Q$ at
$u_{m}$ serves as a useful approximation. Hence we construct the
sequence of corresponding Padé approximants along the main diagonal
of the Padé table. The sequence is terminated when a convergent is
found which satisfies the required accuracy goals. Note like Taylor
polynomials, Padé approximants provide exceptionally good approximations
near a point, in this case the point of expansion $p_{m}$, but the
error deteriorates as we move away from the point. Hence we need only
check the accuracy requirements are met at the end points of the interval
$\left[u_{1},u_{2}\right]$ %
\footnote{This fact is not entirely true, since in some rare cases a Padé approximant
may exhibit spurious poles not present in the original function $Q$.
A more robust algorithm would check that none of the real roots of
the denominator polynomial appearing in the Padé approximant lie in
the interval $\left[u_{1},u_{2}\right]$. Such poles are called defects,
see \cite{jr_pade_2010} for details.%
}. We will discuss how to choose the points $u_{1}$ and $u_{2}$ below.
\item On the left and right regions the left and right solutions of the
recycling equation (\ref{eq:RecyclingEquation}) denoted $A_{L}$
and $A_{R}$ respectively, serve as particularly good approximations.
Note that this is precisely what they were designed to do. Again a
series acceleration technique such as Levin's u-transform \cite{roy_rational_1996}
or Padé summation may be applied to the Taylor polynomials of $A_{L}$
and $A_{R}$. As is usual with such techniques we observed analytic
continuation and increased rates of convergence using these techniques.
The points at which we impose the initial conditions are critical.
For now we have chosen $u_{0}=F_{B}\left(z_{0}\right)$ where $z_{0}=\left(Q_{B}\left(\tau_{L}\right)+Q_{B}\left(u_{m}/2\right)\right)/2$
and $z_{0}=\left(Q_{B}\left(1-\tau_{R}\right)+Q_{B}\left(u_{m}/2\right)\right)/2$
for the left and right problems respectively. But by varying these
initial conditions we alter the range of distribution parameters for
which the algorithm is valid. An optimal choice has yet to be set. 
\end{itemize}
The points $u_{1}$ and $u_{2}$ enclosing the central region are
determined by an estimate $\tilde{r}$ of the radius of convergence
$r$ for the series,

\begin{equation}
Q\left(u\right)=\sum_{n=0}^{\infty}q_{n}(u-u_{0})^{n}.\label{eq:QuantileSeries}
\end{equation}
In particular we set, 

\[
u_{1}=p_{m}-\bar{r},
\]
and,

\[
u_{2}=p_{m}+\bar{r},
\]
where,

\[
\bar{r}=\min\left\{ \tilde{r},\left|p_{m}-0.1\right|,\left|p_{m}-0.9\right|\right\} 
\]

For simplicity our approximation of $\tilde{r}$ will be based on
the Cauchy-Hadamard formula \cite[§ 2.2]{henrici_applied_1974}. However
note that the problem of estimating the radius of convergence is a
rather old one and many more advanced techniques have been developed
to determine $\tilde{r}$. For example Chang and Corliss \cite{chang_ratio-like_1980}
form a small system of equations based on three, four and five terms
of the series (\ref{eq:QuantileSeries}) to determine $\tilde{r}$.
However we are not overly concerned if $\tilde{r}$ over estimates
the radius of convergence of \ref{eq:QuantileSeries} since this will
be compensated by the fact that applying an appropriate summation
technique will provide analytic continuation of (\ref{eq:QuantileSeries}).
Thus for this iteration of the algorithm we will be content with a
simple estimate of $\tilde{r}$ provided by the Cauchy-Hadamard formula. 

We demonstrate the performance of the algorithm by observing a few
test cases for the hyperbolic distribution. Note that it would be
difficult to test the validity of the algorithm for the entire parameter
space of the distribution due to the recursive nature of the coefficients
appearing in the series expansions, so we have biased our testing
to parameters which frequently occur in financial data and a few extreme
cases. As a test case consider the parameters $\alpha=89.72$, $\beta=4.7184$,
$\delta=0.0014$ and $\mu=-0.0015$. These are the estimated parameters
for BMW stock prices reported in \cite{eberlein_hyperbolic_1995},
computed over a 3 year period. Setting the required accuracy to $\epsilon=2.98\times10^{-8}$,
the algorithm constructs an approximant within 0.32 seconds. This
time is for a rather unoptimized Mathematica prototype of the algorithm
running on an Intel i7 laptop. We would expect production code to
be a fraction of this time. The resulting error plots are given in
figure \ref{fig:HyperbolicPadeErrorPlots}. Since the highest degree
of the approximant is only $24$ we expect the algorithm to have reasonably
fast execution times.

\begin{figure}[th]
\hfill{}\subfloat[$y+R_{10,10}\left(1/y\right)$]{\includegraphics[scale=0.85]{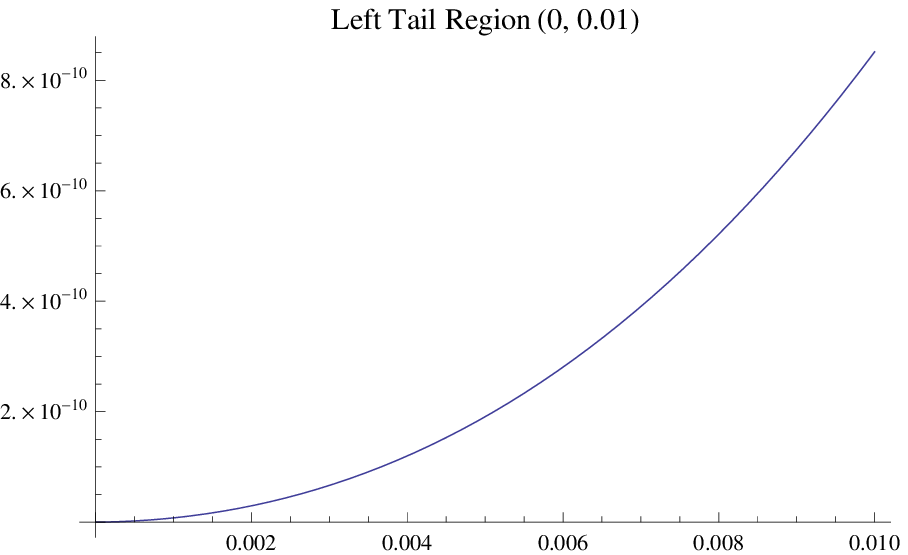}

}\hfill{}\subfloat[$y+R_{10,10}\left(1/y\right)$]{\includegraphics[scale=0.85]{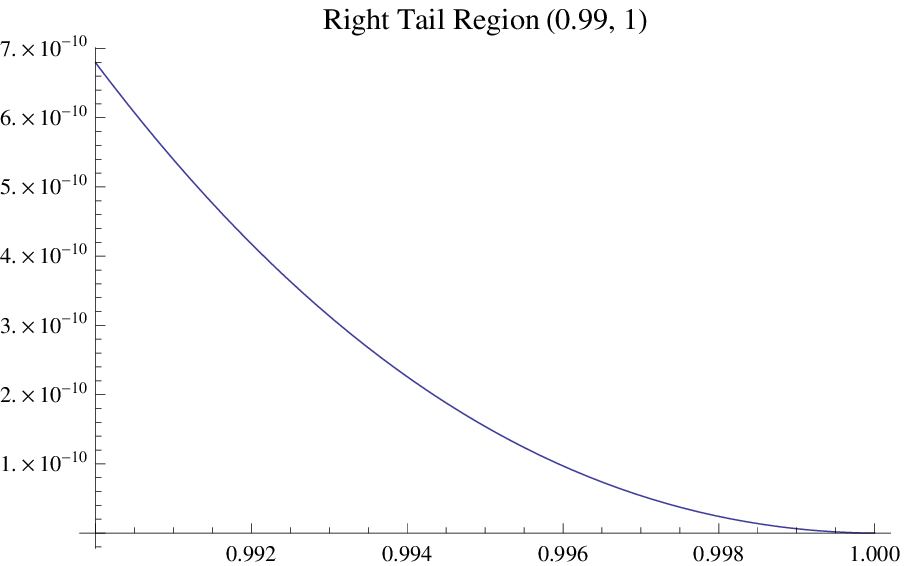}

}\hfill{}

\hfill{}\subfloat[$R_{12,12}\left(z-z_{0}\right)$ ]{\includegraphics[scale=0.85]{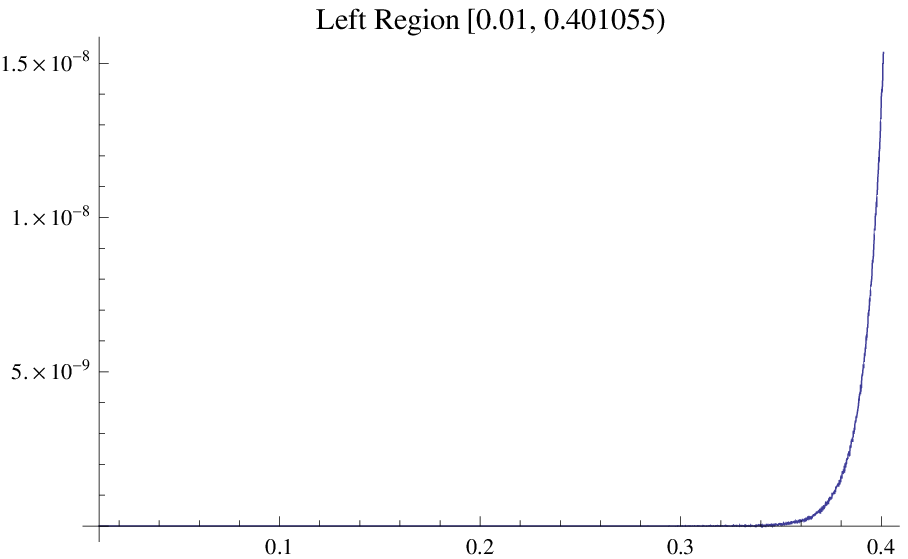}

}\hfill{}\subfloat[$R_{12,12}\left(z-z_{0}\right)$]{\includegraphics[scale=0.85]{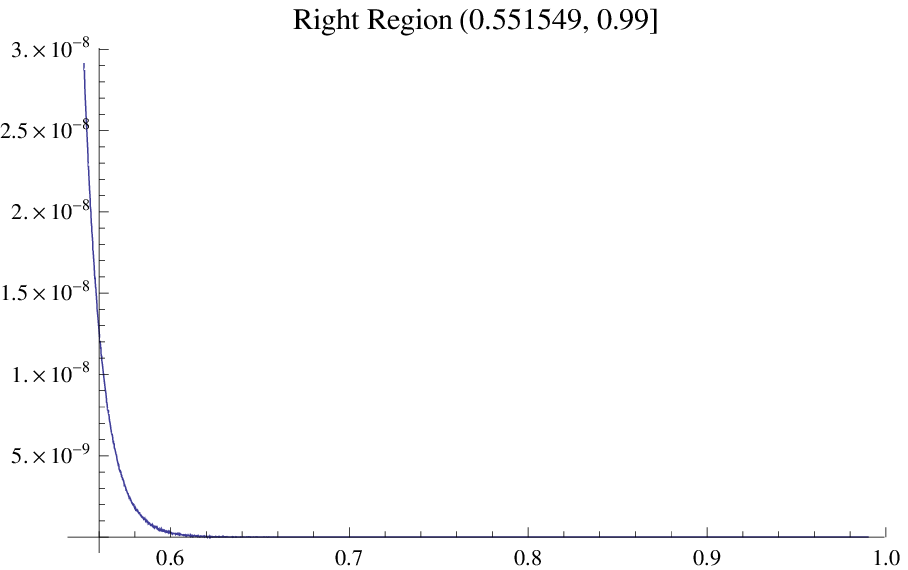}

}\hfill{}

\hfill{}\subfloat[$R_{11,11}\left(u-u_{0}\right)$]{\includegraphics[scale=0.85]{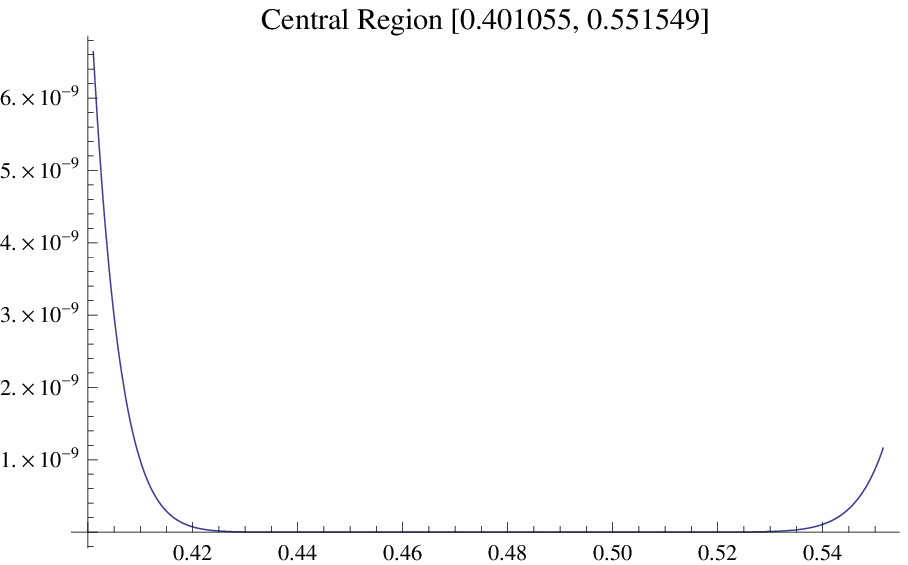}

}\hfill{}

\caption{\label{fig:HyperbolicPadeErrorPlots}Algorithm 1 (Padé Approximants):
Error Plots $\left|u-F\left(Q_{A}\left(u\right)\right)\right|$. }
\end{figure}

Despite meeting the accuracy requirements for the given set of parameters,
as can be seen from figure \ref{fig:HyperbolicPadeErrorPlots} the
error curve produced by this algorithm is far from optimal in the
minimax sense. Thus our next two algorithms have to do with constructing
so called near best approximants, which are often good enough in practice
due to the difficulties involved in finding $R_{m,n}^{*}$. To this
end consider the partition $\tau_{L}<u_{1}<u_{2}<1-\tau_{R}$. For
simplicity we choose $\tau_{L}=\tau_{R}=10^{-10}$. We are interested
in the problem of constructing the Chebyshev series expansions of
the functions $A_{L}\left(z\left(u\right)\right)$, $Q\left(u\right)$
and $A_{R}\left(z\left(u\right)\right)$ restricted to the sets $\left[\tau,u\right]$,
$\left[u_{1},u_{2}\right]$ and $\left[u_{2},1-\tau_{R}\right]$ respectively.
In each case we introduce a linear change of variable $x$ which maps
the restricted domain onto the set $\left[-1,1\right]$. To ease notation
let $g\left(x\right):=\left.Q\right|_{\left[u_{1},u_{2}\right]}\left(x\right)$;
in this case $\left.x\right|\left[u_{1},u_{2}\right]\rightarrow\left[-1,1\right]$
is defined by, 

\begin{equation}
x\left(u\right)=\frac{u-\frac{1}{2}(u_{2}+u_{1})}{\frac{1}{2}(u_{2}-u_{1})}=\frac{u-p_{m}}{\bar{r}}.\label{eq:xU}
\end{equation}
Following from the properties of $Q$, the function $g$ is continuous
and of bounded total variation and thus admits the expansion, 

\begin{equation}
g\left(x\right)=\frac{\tilde{g}_{0}}{2}+\sum_{k=0}^{\infty}\tilde{g}_{k}T_{k}\left(x\right),\quad x\in\left(0,1\right),\label{eq:ChebyshevSeriesQ}
\end{equation}
where $T_{k}\left(x\right)$ are the Chebyshev polynomials of the
first kind and the coefficients $c_{k}$ are defined by,

\begin{equation}
\tilde{g}_{k}=\frac{2}{\pi}\int_{-1}^{1}g\left(x\right)\frac{T_{k}\left(x\right)}{\sqrt{1-x^{2}}}dx.\label{eq:ChebyshevCoefficientsQ}
\end{equation}
Since the quantile functions $Q$ we have considered in this report
are assumed to be infinitely differentiable, elementary Fourier theory
tells us that the error made in truncating the series (\ref{eq:ChebyshevSeriesQ})
after $K$ terms goes to zero more rapidly than any finite power of
$1/K$ as $K\rightarrow\infty$, \cite[§ 3]{gottlieb_numerical_1977}.
Such rapid decrease of the remainder motivates us to seek efficient
methods to evaluate the integral in (\ref{eq:ChebyshevCoefficientsQ}).
This task can rarely be performed analytically so the the usual process
is to write (\ref{eq:ChebyshevCoefficientsQ}) as,

\[
\tilde{g}_{k}=\frac{2}{\pi}\int_{0}^{\pi}g\left(\cos\theta\right)\cos k\theta d\theta,
\]
and apply a variant of the fast Fourier transform \cite[§ 29.3]{hamming_numerical_1973}.
This method however would require several thousand evaluations of
the function $g$ and result in unacceptable setup times. Fortunately
an alternative approach is provided by Thacher \cite{thacherjr._conversion_1964}.
Let the Taylor series expansion of $g$ be given by, 

\begin{equation}
g\left(x\right)=\sum_{k=0}^{\infty}g_{k}x^{k}.\label{eq:gTaylorSeries}
\end{equation}
Then by inverting (\ref{eq:xU}) and substituting into (\ref{eq:QuantileSeries})
we see that the coefficients appearing in (\ref{eq:gTaylorSeries})
may be written as $g_{k}=\bar{r}q_{k}$. Now substituting the following
relationship between the monomials $x^{k}$ and the Chebyshev polynomials
\cite[eq. 4]{thacherjr._conversion_1964},

\[
x^{k}=T_{0}\left(x\right)+\sum_{j=1}^{k}\theta_{k,j}T_{j}\left(x\right),
\]
where,

\[
\theta_{j,k}=\begin{cases}
2^{1-j}\binom{j}{\frac{j-k}{2}} & j-k\;\textrm{even}\\
0 & j-k\;\textrm{odd}
\end{cases},
\]
into (\ref{eq:gTaylorSeries}), one may express the Chebyshev coefficients
$\tilde{g}_{k}$ in terms of the Taylor coefficients $g_{k}$ as follows, 

\begin{equation}
\tilde{g}_{k}=\sum_{j=k}^{\infty}g_{j}\theta_{j,k},\label{eq:ChebyTaylorCoeffRelationship}
\end{equation}
Thacher's approach was simple, he observed that applying Shank's transform
to (\ref{eq:ChebyTaylorCoeffRelationship}) one may approximate the
Chebyshev coefficients even for slowly convergent series. In our experimentation
we found Levin's $u$-transform to also be affective, which of course
was not discovered at the time Thacher wrote his paper. Note that
Levin's $u$ transform has been reported in many instances to outperform
Shank's transform \cite{smith_numerical_1982-1}. For efficient Fortran
and C++ implementations of these two transforms see \cite{weniger_nonlinear_2001}
and \cite{press_numerical_2007} respectively. Thus given the knowledge
of the Taylor coefficients $g_{k}$ we now have a method to approximate
the Chebyshev coefficients without having to resort to the computationally
expensive task of evaluating (\ref{eq:ChebyshevCoefficientsQ}) through
numerical quadrature. An estimate of the error made in truncating
the series (\ref{eq:ChebyshevSeriesQ}) is given by the magnitude
of the first neglected term. In a similar fashion we can construct
Chebyshev series expansions for the restrictions $\left.A_{L}\right|_{\left[\tau_{L},u_{1}\right]}$
and $\left.A_{R}\right|_{\left[u_{2},1-\tau_{R}\right]}$. As can
be seen from figure \ref{fig:errorPlotTruncatedChebyshev} an algorithm
based on constructing truncated Chebyshev expansions in this manner
provides us with a much more satisfactory error curve. 

\begin{figure}[th]
\hfill{}\subfloat[$T_{34,0}\left(x\right)$]{\includegraphics[scale=0.85]{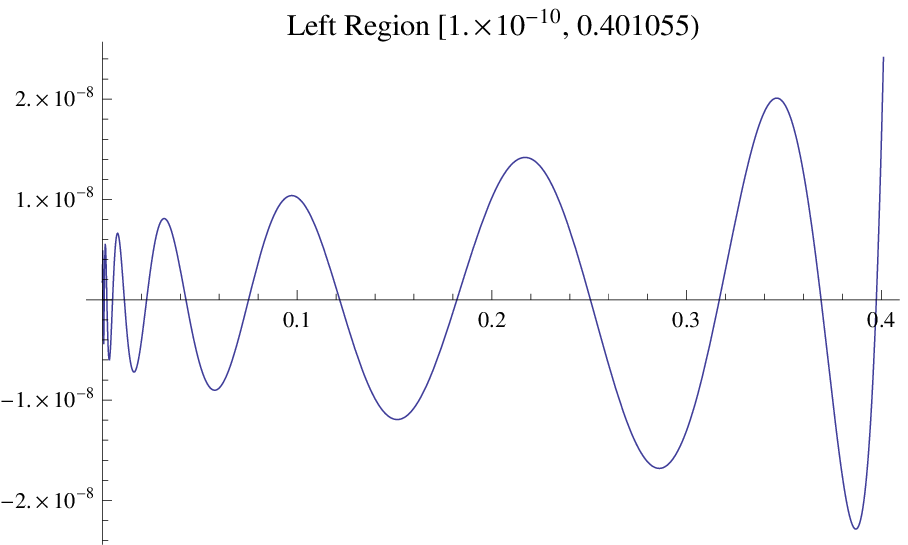}

}\hfill{}\subfloat[$T_{35,0}\left(x\right)$]{\includegraphics[scale=0.85]{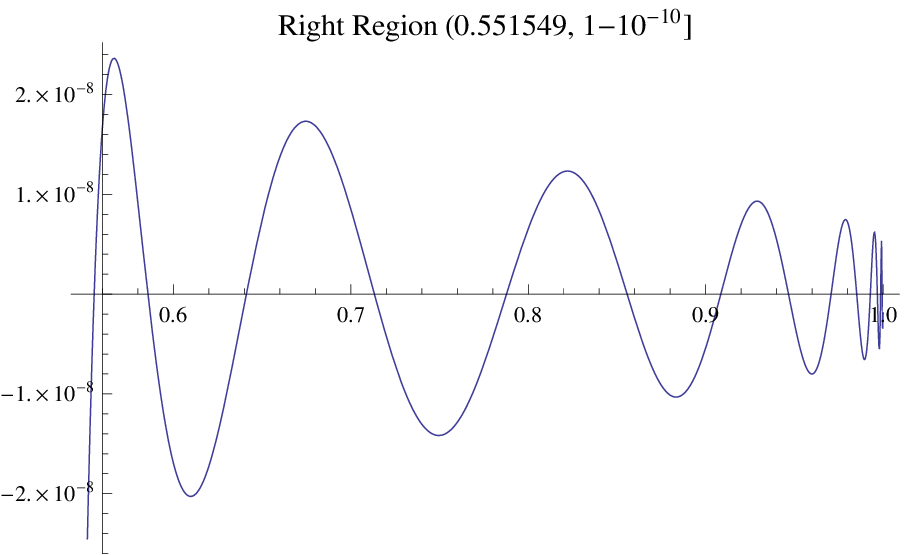}

}\hfill{}

\hfill{}\subfloat[$T_{7,0}\left(x\right)$]{\includegraphics[scale=0.85]{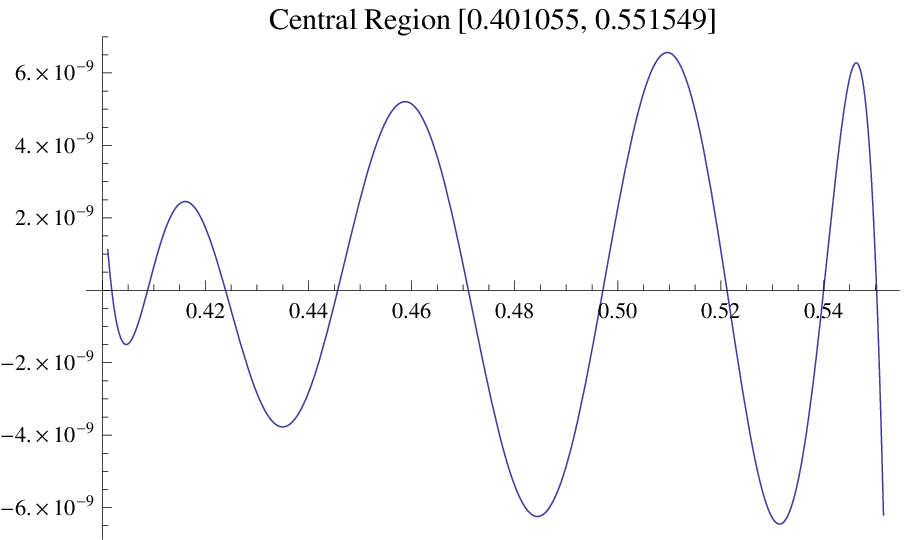}

}\hfill{}

\caption{\label{fig:errorPlotTruncatedChebyshev}Algorithm 2 (Truncated Chebyshev
Series): Error Plots $\left(Q\left(u\right)-Q_{A}\left(u\right)\right)$}
\end{figure}

Perhaps the greatest drawback of this algorithm is that it requires
one to compute a large number of Chebyshev coefficients. In our example
the series needed to be truncated at the $35^{\textrm{th}}$ term
in order to satisfy the accuracy requirements in the right region,
see figure \ref{fig:errorPlotTruncatedChebyshev}. We stress that
approximating these $35$ coefficients via \ref{eq:ChebyTaylorCoeffRelationship}
is not a computationally expensive task in contrast to a numerical
integration scheme, however evaluating a $35^{\textrm{th}}$ degree
polynomial may lead to unacceptable execution times. It is for this
reason we consider constructing Chebyshev-Padé approximants of the
function $g$, which are rational functions of the form, 

\begin{equation}
T_{m,n}\left(x\right):=\frac{a_{0}T_{0}\left(x\right)+\cdots+a_{m}T_{m}\left(x\right)}{b_{0}T_{0}\left(x\right)+\cdots+b_{n}T_{n}\left(x\right)},\label{eq:PadeChebyshev}
\end{equation}
which have a formal Chebyshev series expansion in agreement with (\ref{eq:ChebyshevSeriesQ})
up to and including the term $\tilde{g}_{m+n}T_{m+n}\left(x\right)$.
Given the knowledge of the first $n+m+1$ Chebyshev coefficients $\tilde{g}_{0},\ldots,\tilde{g}_{n+m}$,
an efficient way to construct $T_{m,n}\left(x\right)$ is to employ
the algorithm of Sidi \cite{avram_computation_1975}, in which the
coefficients appearing in \ref{eq:PadeChebyshev} are computed recursively.
For an estimate of the error given by (\ref{eq:PadeChebyshev}) see
\cite[eq. 7.6-12]{ralston_first_2001}. Applying this scheme to the
same test case as above we obtain the error plots given in figure
(\ref{fig:ErrorPlotsChebyshevPade}).

\begin{figure}[h]
\hfill{}\subfloat[$T_{4,4}\left(x\right)$ ]{\includegraphics[scale=0.85]{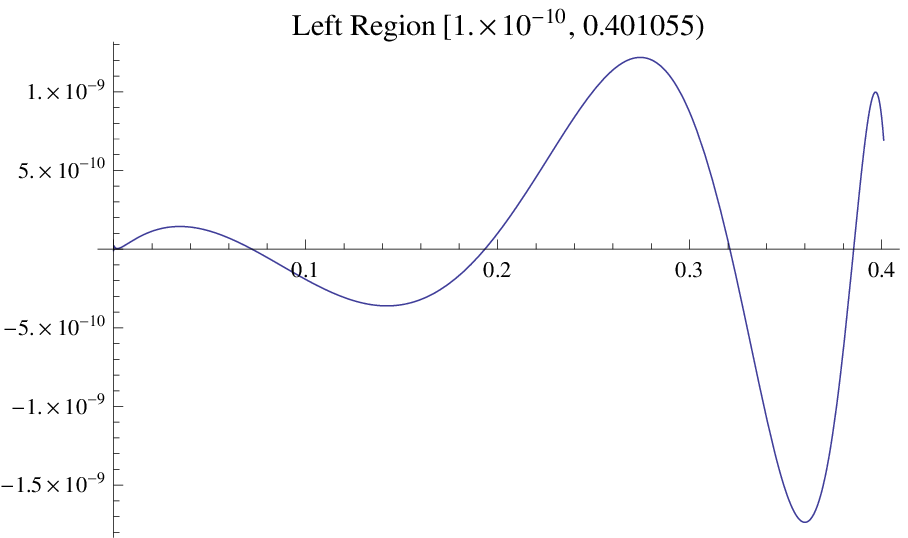}

}\hfill{}\subfloat[$T_{6,6}\left(x\right)$]{\includegraphics[scale=0.85]{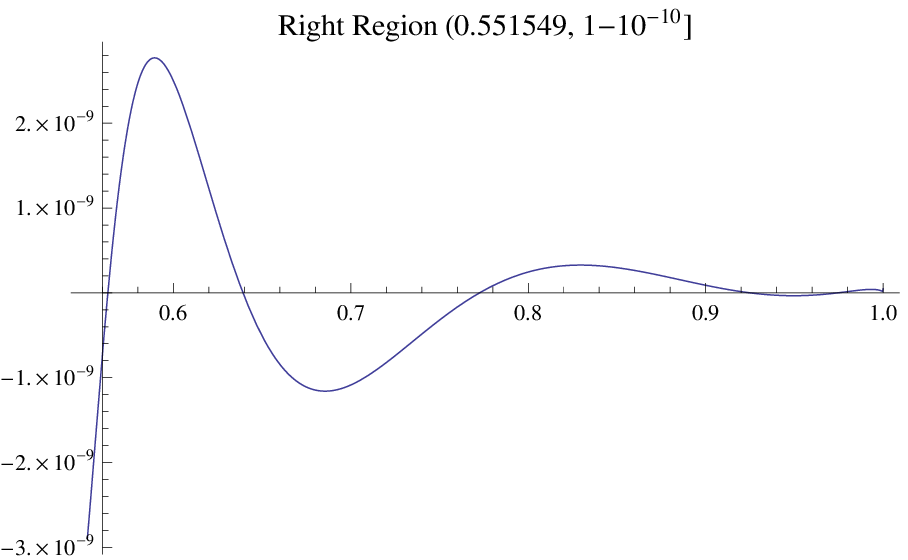}

}\hfill{}

\hfill{}\subfloat[$T_{4,4}\left(x\right)$ ]{\includegraphics[scale=0.85]{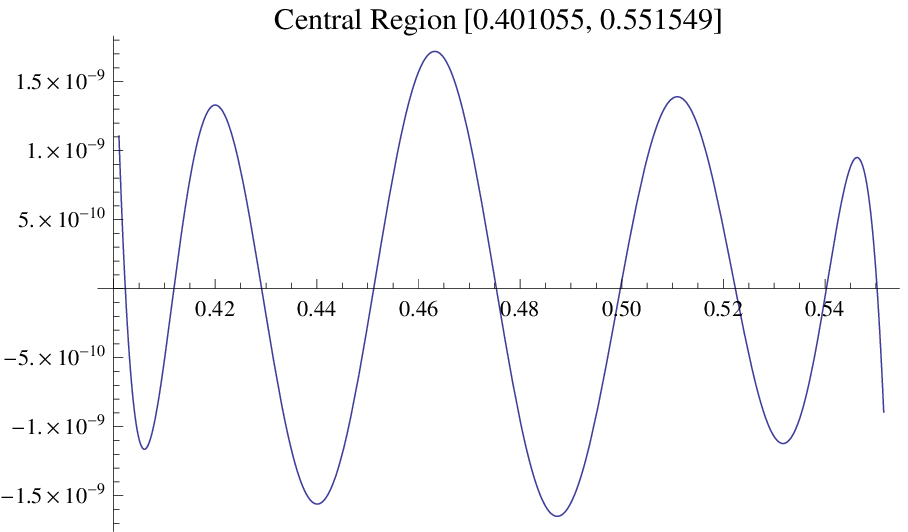}

}\hfill{}

\caption{\label{fig:ErrorPlotsChebyshevPade}Algorithm 3 (Chebyshev-Padé Approximants
): Error Plots $\left(Q\left(u\right)-Q_{A}\left(u\right)\right)$}
\end{figure}

The real benefit of this algorithm is the fast setup and execution
times it provides. For example we required a truncated Chebyshev series
of degree $34$ to approximate the hyperbolic quantile to the required
accuracy in the left region of our test case, again see figure (\ref{fig:errorPlotTruncatedChebyshev}).
On the other hand a Chebyshev-Padé approximant of degree only $8$
provides us with a much better approximant in the same region, contrast
with figure (\ref{fig:ErrorPlotsChebyshevPade}). 

Our final algorithm is based on multipoint Padé approximants \cite{jr_pade_2010}.
Consider the partition $\tau_{L}<p_{m}<1-\tau_{R}$, where $p_{m}$
is the mode location. Choose the points $u_{0},\ldots,u_{K}\in\left[\tau_{L},p_{m}\right]$,
and let $s_{0},\ldots,s_{K}\in\mathbb{N}$. Suppose now we want to
build a rational approximant $R_{m,n}$ valid on $\left[\tau_{L},p_{m}\right]$
which satisfies the following conditions, 

\begin{equation}
A_{L}^{\left(j\right)}\left(z_{k}\right)=R_{m,n}^{\left(j\right)}\left(z_{k}\right),\quad0\leq k\leq K,\;0\leq j\leq s_{k}-1,\label{eq:InterpolationData}
\end{equation}
where $z_{k}=Q_{B}\left(u_{k}\right)$. Such a problem is called an
osculatory rational interpolation problem. The interpolation data
on the right hand side of (\ref{eq:InterpolationData}) may be generated
by solving the recycling equation (\ref{eq:RecyclingEquation}) for
$A_{L}$ with initial conditions imposed at the points $u_{0},\ldots,u_{K}$.
Provided a solution to the problem exists it may be solved efficiently
through the generalized Q.D. algorithm \cite{p.r._generalised_1980}.
This algorithm generates the partial numerators and denominators of
a continued fraction, who's convergents are rational functions called
multipoint Padé approximants satisfying (\ref{eq:InterpolationData}).
Similarly by considering solutions to $A_{R}$ we may find a rational
approximant valid on $\left[p_{m},1-\tau_{R}\right]$. By choosing
only $3$ interpolation points, we apply the method to our test case
and provide the results in figure (\ref{fig:MultipointPadeErrorPlot}). 

\begin{figure}
\hfill{}\subfloat[$R_{10,10}\left(z\right)$. Interpolation Data: $u_{0}=p_{m}/3$,
$u_{1}=p_{m}$, $s_{0}=15$, $s_{1}=6$]{\includegraphics[scale=0.85]{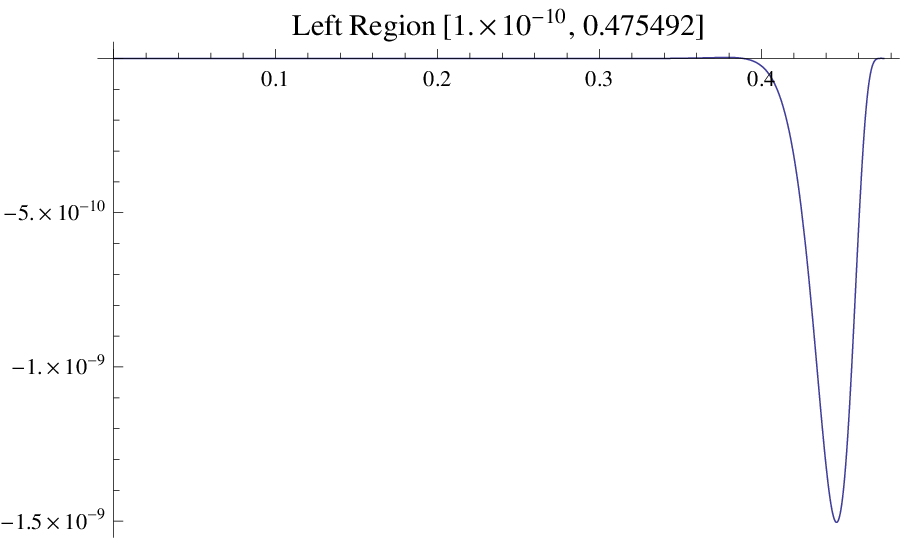}

}\hfill{}\subfloat[$R_{10,10}\left(z\right)$. Interpolation Data: $u_{0}=p_{m}$, $u_{1}=\left(2+p_{m}\right)/3$,
$s_{0}=5$, $s_{1}=16$]{\includegraphics[scale=0.85]{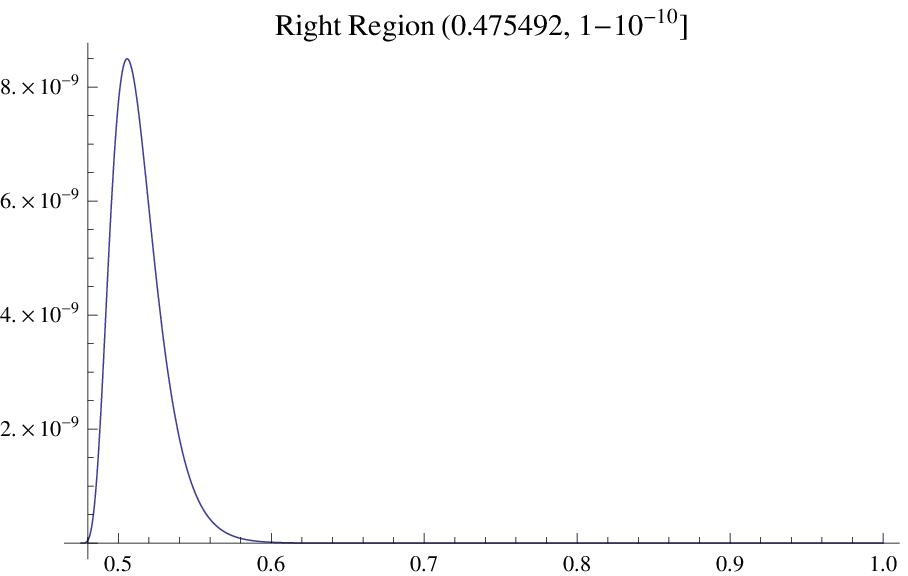}

}\hfill{}

\caption{\label{fig:MultipointPadeErrorPlot}Algorithm 4 (Multipoint Padé Approximants):
Error Plots $u-F\left(Q_{A}\left(u\right)\right)$}
\end{figure}

One could go a step further and attempt to construct the minimax approximant
$R_{m,n}^{*}$, however as mentioned most of the methods known to
the present author to construct $R_{m,n}^{*}$ require evaluating
$Q$ at a large number of points. It is for this reason we recommend
Maehly's indirect method \cite{maehly_methods_1963}. This method
of constructing $R_{m,n}^{*}$ is applicable when a sufficiently accurate
approximation of $Q$ is available. Such an initial approximation
may be obtained by the methods described above. Maehly's indirect
method then seeks to find a more optimal approximation in the Chebyshev
sense without the need to evaluate $Q$, which of course in our case
is an expensive operation. Thus one would expect reduced setup up
times with this approach, however we have not explored this idea further. 

\bibliographystyle{plain}
\bibliography{C:/WorkArea/PhD/Bibtex/ReferenceLibrary}

\end{document}